\documentclass[journal=jctcce,manuscript=article,layout=twocolumn]{achemso}

\usepackage{graphicx}
\usepackage{dcolumn}
\usepackage{bm}
\usepackage{graphicx,amsmath}
\usepackage{tikz}
\usepackage[version=4]{mhchem} 
\usetikzlibrary{quantikz}
\usepackage{subfigure}
\usepackage{amsmath}
\usepackage{amssymb}
\usepackage{algpseudocode}
\usepackage{algorithm}
\usepackage{amsthm}

\usepackage{makecell}
\usepackage{physics}
\usepackage{dblfloatfix}
\usepackage{placeins}
\usepackage{xurl}

\newtheorem{lem}{Lemma}
\newtheorem{thm}{Theorem}

\title{
Clifford disentanglers for entanglement reduction in molecular electronic structure simulations
}
\author{Longfei Chang}\altaffiliation{Contributed equally to this work.}
\affiliation{Key Laboratory of Theoretical and Computational Photochemistry, Ministry of Education, College of Chemistry, Beijing Normal University, Beijing 100875, China}
\alsoaffiliation{Institute for Advanced Study, Beijing Normal University, Beijing, 100875, China}

\author{Zibo Wu}\altaffiliation{Contributed equally to this work.}
\affiliation{Key Laboratory of Theoretical and Computational Photochemistry, Ministry of Education, College of Chemistry, Beijing Normal University, Beijing 100875, China}
\alsoaffiliation{Institute for Advanced Study, Beijing Normal University, Beijing, 100875, China}

\author{Yunzhi Li}\altaffiliation{Contributed equally to this work.}
\affiliation{Key Laboratory of Theoretical and Computational Photochemistry, Ministry of Education, College of Chemistry, Beijing Normal University, Beijing 100875, China}
\alsoaffiliation{Institute for Advanced Study, Beijing Normal University, Beijing, 100875, China}

\author{Haiqi Liu}
\affiliation{Key Laboratory of Theoretical and Computational Photochemistry, Ministry of Education, College of Chemistry, Beijing Normal University, Beijing 100875, China}
\alsoaffiliation{Institute for Advanced Study, Beijing Normal University, Beijing, 100875, China}

\author{Jiajun Ren}\email{jjren@bnu.edu.cn}
\affiliation{Key Laboratory of Theoretical and Computational Photochemistry, Ministry of Education, College of Chemistry, Beijing Normal University, Beijing 100875, China}
\alsoaffiliation{Institute for Advanced Study, Beijing Normal University, Beijing, 100875, China}

\author{Mingpu Qin}\email{qinmingpu@sjtu.edu.cn}
\affiliation{Key Laboratory of Artificial Structures and Quantum Control (Ministry of Education),
School of Physics and Astronomy, Shanghai Jiao Tong University, Shanghai 200240, China}
\alsoaffiliation{Hefei National Laboratory, Hefei 230088, China}

\author{Zhendong Li}\email{zhendongli@bnu.edu.cn}
\affiliation{Key Laboratory of Theoretical and Computational Photochemistry, Ministry of Education, College of Chemistry, Beijing Normal University, Beijing 100875, China}
\alsoaffiliation{Institute for Advanced Study, Beijing Normal University, Beijing, 100875, China}

\author{Wei-Hai Fang}
\affiliation{Key Laboratory of Theoretical and Computational Photochemistry, Ministry of Education, College of Chemistry, Beijing Normal University, Beijing 100875, China}
\alsoaffiliation{Institute for Advanced Study, Beijing Normal University, Beijing, 100875, China}

\begin{document}

\begin{tocentry}

\includegraphics[width=\textwidth]{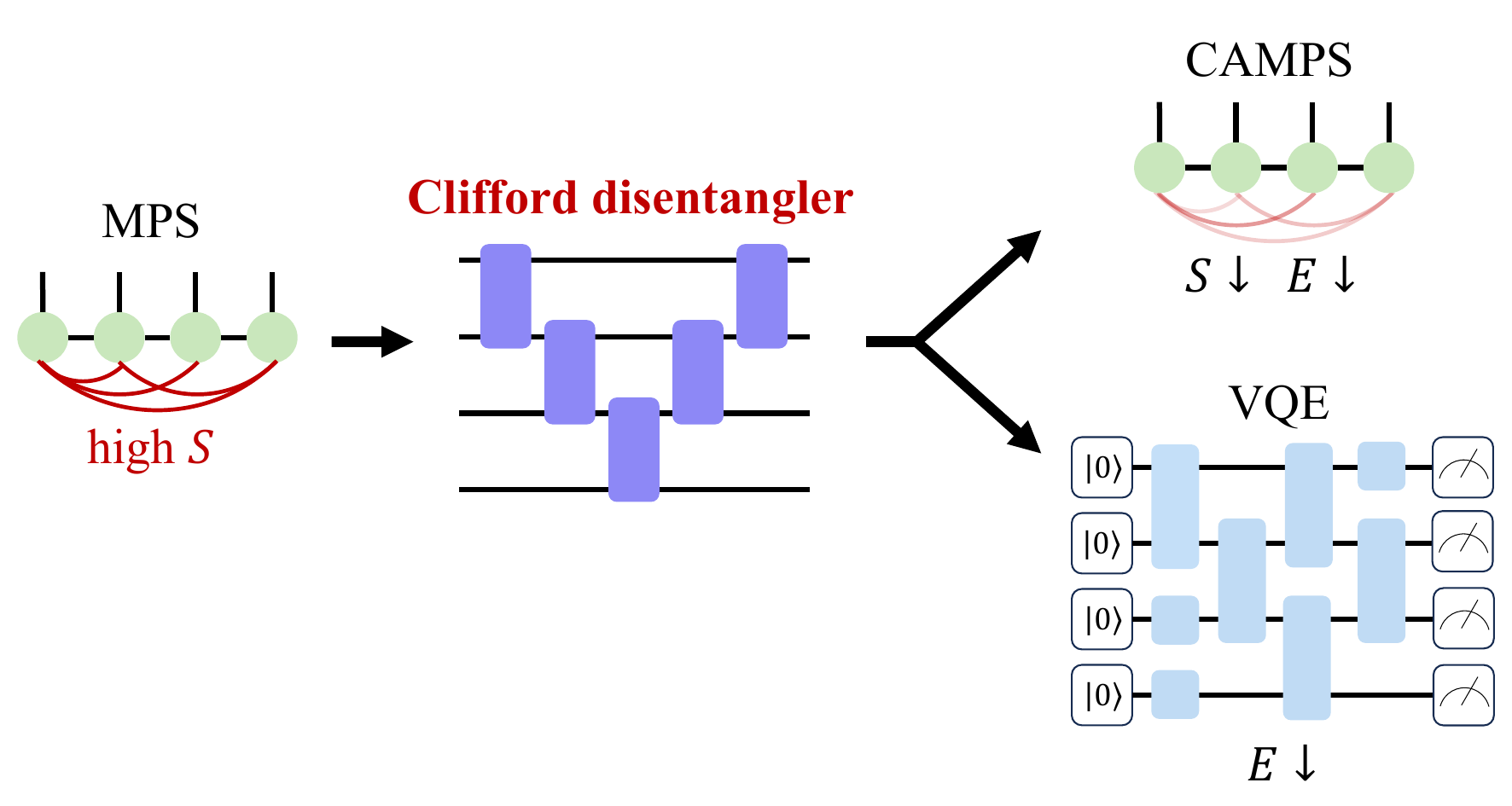}





\end{tocentry}

\begin{abstract}
Entanglement is a key bottleneck limiting the efficiency of tensor-network and quantum simulations of molecular electronic structures. Here, we systematically assess and extend Clifford disentanglers as a structure-preserving approach to entanglement reduction: they can modify the entanglement structure of qubit wavefunctions while retaining the Pauli-string form of qubit Hamiltonians. To enable a practical search over Clifford transformations, we classify Clifford operators by their action on the Schmidt spectrum across a bipartition, reducing the two- and four-qubit search spaces to 20 and 91392 representatives, respectively. Embedded in an iterative Clifford-augmented matrix product state framework, these transformations reduce the energy errors at fixed bond dimension for the molecular test cases studied and mitigate the dependence on orbital orderings and fermion-to-qubit mappings. We further show that Clifford disentanglers can also benefit quantum simulations such as the shallow-circuit variational quantum eigensolver calculations. 
Together, these results establish Clifford disentanglers as a useful structure-preserving entanglement-engineering tool for tensor-network and quantum simulations of molecular electronic structure, while also clarifying their correlation dependence and
motivating future developments.
\end{abstract}

\maketitle

\section{Introduction}
Entanglement fundamentally governs the computational complexity of quantum many-body simulations.\cite{vidal2003efficient, eisert2010colloquium} In tensor network approaches, the cost of representing a quantum state scales with its bipartite entanglement entropy rather than the dimension of the Hilbert space.\cite{orus2014practical, cirac2021matrix,xiang2023density} The accumulation of short-range entanglement was shown to hinder conventional real-space renormalization schemes, motivating Vidal to introduce the disentanglers in the multiscale entanglement renormalization ansatz (MERA)\cite{vidal_entanglement_2007,vidal_class_2008}. By applying local unitary transformations prior to coarse-graining, disentanglers systematically reshape entanglement structure. Beyond renormalization group settings, disentangling transformations have also been incorporated into matrix product state (MPS) frameworks\cite{white_density_1992,eisert2010colloquium,ran2020encoding}. Matrix product disentanglers have been proposed for state preparation\cite{ran2020encoding}, while the fully-augmented MPS (FAMPS) approach introduces them to enhance the expressive power of MPS for quantum many-body problems\cite{qian2023augmenting}. In the latter case, determining the disentanglers requires variational optimization using the standard Evenbly–Vidal algorithm\cite{evenbly2009algorithms}, which introduces additional computational complexity. This complexity motivates the search for more structured and analytically tractable classes of disentangling transformations.

In this context, Clifford operators can serve as a better candidate for disentanglers.\cite{gottesman1998heisenberg, nielsen2010quantum} Composed of Hadamard (H), phase (S), and controlled-NOT (CNOT) gates, these operators can generate substantial entanglement while remaining classically simulable by virtue of the Gottesman–Knill theorem\cite{gottesman1998heisenberg, aaronson2004improved}. More importantly, the conjugation of a qubit Hamiltonian, expressed as a linear combination of Pauli operators $\hat{H}=\sum_{i} c_i P_i$, by a Clifford operator maps each Pauli term to another Pauli term, leaving the total number of terms in the Hamiltonian representation unchanged\cite{nielsen2010quantum, bravyi2002fermionic}. This key property ensures that Clifford disentanglers can restructure the entanglement patterns of many-body wavefunctions without significantly increasing the complexity of variational optimization.
Leveraging these advantages, several Clifford-enhanced quantum or classical algorithms have been developed for many-body problems. The Schrödinger–Heisenberg variational quantum algorithm (SHVQA) integrates classically simulable Heisenberg circuits to enhance the expressivity of local unitary ansätze\cite{shang2023schrodinger}. Hierarchical Clifford transformation (HCT) methods exploit approximate $\mathbb{Z}_2$  Pauli symmetries to achieve near-block-diagonalization of Hamiltonians, reducing bipartite entanglement entropy and improving convergence in variational quantum eigensolvers\cite{mishmash2023hierarchical}. The Clifford-based Hamiltonian engineering method (CHEM) optimizes the Hamiltonian representation to enlarge the initial energy gradient of hardware-efficient ansätze, mitigating optimization difficulties on near-term quantum devices\cite{sun2024toward}. More recently, the Clifford-augmented matrix product states  (CAMPS) algorithm has employed $2$-qubit Clifford operators to minimize the singular values discarded during the sweeping procedure of DMRG calculations\cite{qian_augmenting_2024,huang_clifford_2025,qian_clifford_2025,fu_clifford_2025,huang2026augmenting}. Collectively, these developments establish Clifford operators as an important tool at the intersection of classical and quantum simulation paradigms for quantum many-body problems.

In this work, we systematically assess and extend Clifford disentanglers for electronic structure problems. We construct disentanglers from MPS wavefunctions and examine how the resulting Clifford-transformed Hamiltonians affect subsequent 
density matrix renormalization group (DMRG)\cite{white_density_1992} and
variational quantum eigensolver (VQE)\cite{peruzzo2014variational,mcclean2016theory} 
calculations. Specifically, we assess the accuracy improvement, the reduction of bipartite entanglement, and the dependence of CAMPS calculations on orbital orderings, fermion-to-qubit mappings, basis sets, and disentangling-unit size (2-qubit and 4-qubit) based on a classification of Clifford operators utilizing a symplectic representation.
This study delineates both the benefits and limitations of Clifford disentanglers 
as an entanglement-engineering strategy for molecular simulations.

\section{Theory and algorithms}
\subsection{Electronic structure problems and matrix product states}
We consider electronic structure problems described by the following second-quantized Hamiltonian 
\begin{gather}
\hat{H} =
\sum_{pq,\sigma}h_{pq}\hat{a}_{p\sigma}^\dagger\hat{a}_{q\sigma}
+\frac12\sum_{pqrs,\sigma\tau}g_{pqrs}\hat{a}_{p\sigma}^\dagger\hat{a}_{r\tau}^\dagger\hat{a}_{s\tau}\hat{a}_{q\sigma},
\end{gather}
where $h_{pq}$ and $g_{pqrs}$ are one- and two-electron integrals in an orthonormal molecular orbital basis, and $\hat a^\dagger_{p\sigma}$ ($\hat a_{p\sigma}$) are fermionic creation (annihilation) operators.
In this work, the fermionic Hamiltonian is mapped to a qubit Hamiltonian using standard fermion-to-qubit transformations,
such as the Jordan--Wigner (JW),\cite{jordan1928pauli}
parity,\cite{whitfield2011simulation,seeley2012bravyi}
and Bravyi--Kitaev (BK) mappings\cite{bravyi2002fermionic,seeley2012bravyi}
for the subsequent tensor network and quantum circuit simulations.

For a system with $K$ spin-orbitals, 
the MPS wavefunction is written as\cite{chan2016matrix}
\begin{align}
\ket{\Psi}
=
\sum_{s_1\cdots s_K}
A^{s_1}[1]
A^{s_2}[2]
\cdots
A^{s_K}[K]
\ket{s_1\cdots s_K},
\end{align}
where $A^{s_k}[k]$ is an $M\times M$ matrix, while the leftmost and rightmost sites 
are $1\times M$ and $M\times 1$ vectors, respectively.
Increasing $M$ enlarges the variational space and allows the MPS to represent states with larger entanglement. The Hamiltonian can be expressed as a matrix product operator (MPO), which enables efficient evaluation of expectation values and local tensor optimization by minimizing the energy of MPS. In this work, the Clifford disentanlers will be derived from a low-bond-dimension MPS (vide post) obtained from the DMRG algorithm.

\begin{figure*}[!t]
    \centering
    \includegraphics[width=\textwidth]{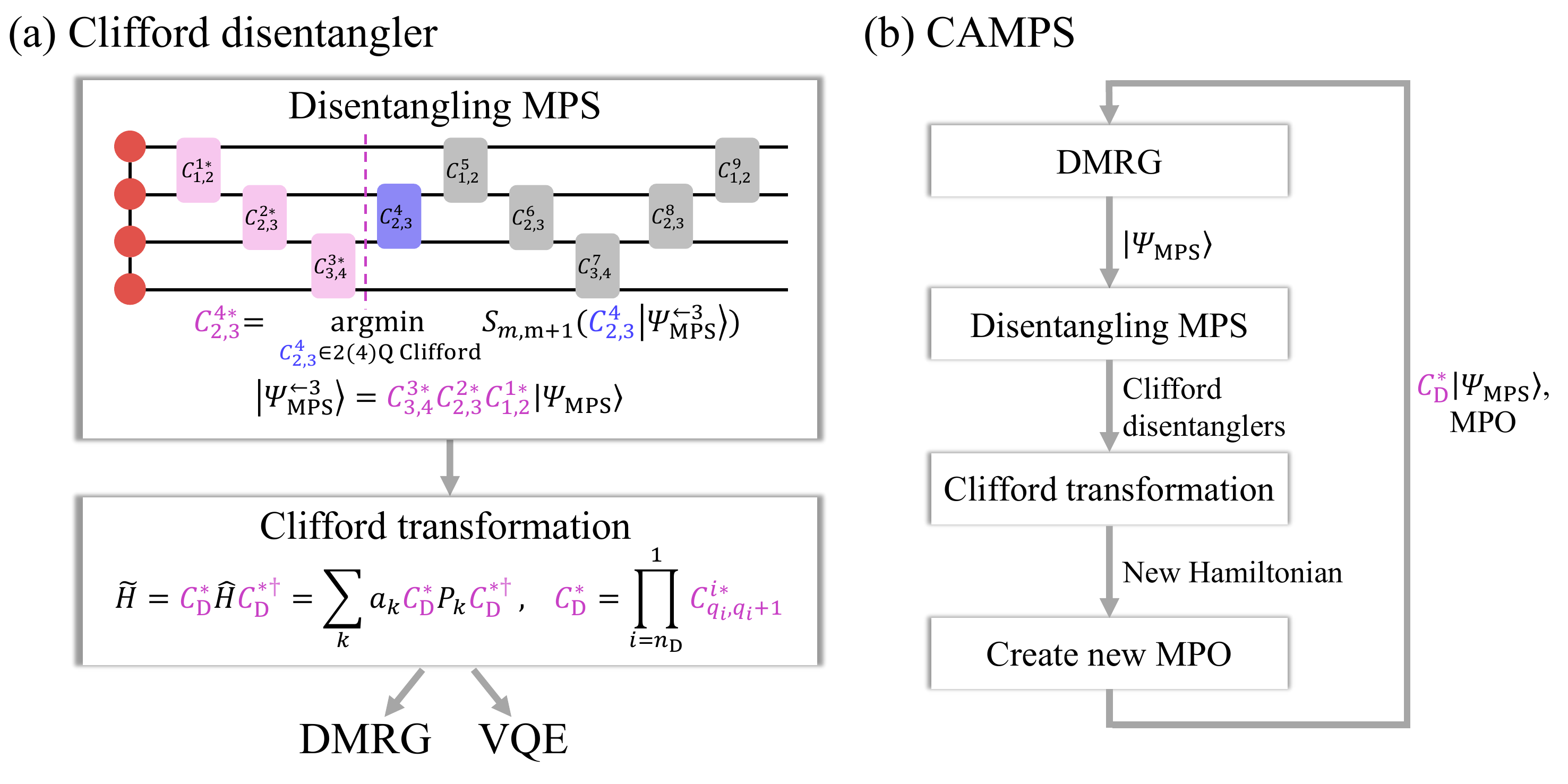}
    \caption{Schematic illustration of the determination of Clifford disentanglers from an MPS and the applications in DMRG and VQE. 
    (a) Local 2Q or 4Q Clifford gates are applied sequentially to an MPS to reduce entanglement. 
    For each bond, representatives of the relevant equivalent classes are tested
to minimize the bipartite entropy $S_{m,m+1}$. The sweeping process is repeated until the entropy converges, yielding the optimized disentangler $C_D^*$ within this
local search protocol. The transformed Hamiltonian can be used in subsequent DMRG and VQE.
(b) In CAMPS, DMRG first produces an MPS $|\Psi_{\mathrm{MPS}}\rangle$; then $C_D^*$ is determined
from the obtained MPS and the disentangled MPS $C_D^*|\Psi_{\mathrm{MPS}}\rangle$
is truncated to the original bond dimension.
The corresponding transformed Hamiltonian is converted to an MPO using
the bipartite graph algorithm in Ref. \cite{ren2020general}
for the next DMRG optimization. These steps are iterated until the energy converges.
}\label{fig: Clifford_disentangler}
\end{figure*}

\subsection{Clifford Operators and Symplectic Representation}
The Pauli group $\mathcal{P}_n$ on $n$ qubits is generated by tensor products of 
single-qubit Pauli operators $\{X,Y,Z\}$ with phase factors.
Specifically, any $n$-qubit Pauli operator can be written as
\begin{equation}
P = i^\kappa \bigotimes_{j=1}^n X_j^{x_j} Z_j^{z_j},
\end{equation}
where $x_j,z_j\in\mathbb Z_2$ and $\kappa = 0, 1, 2, 3$.
The Clifford group is defined as the normalizer of the Pauli group,
\begin{equation}
C\mathcal P_n C^\dagger = \mathcal P_n.
\end{equation}
Therefore, conjugation by a Clifford operator maps a Pauli operator to another Pauli operator.

Because of this property, Clifford circuits acting on the reference state $\ket{00\cdots 0}$ generate stabilizer states, which remain efficiently classically simulable according to the Gottesman–Knill theorem.\cite{gottesman1998heisenberg}
In the stabilizer formalism, an $n$-qubit stabilizer state is specified by $n$ independent commuting Pauli generators, which can be updated algebraically under Clifford circuits.\cite{gottesman1997stabilizer,aaronson2004improved}

This evolution admits a binary symplectic representation,
in which an $n$-qubit Clifford operator is represented by
a $2n\times2n$ matrix $S$ over $\mathbb Z_2$ satisfying
\begin{align}
S \Omega S^T = \Omega,
\qquad
\Omega = \bigoplus_{i=1}^n
\begin{pmatrix}
0&1\\
1&0
\end{pmatrix},
\end{align}
where each $S$ matrix encodes a mapping that maps the $n$-qubit Pauli group to itself (ignoring the phase):
\begin{align}
    X_i \to & \prod_{j=1}^n X_j^{S_{(2i-1)(2j-1)}} Z_j^{S_{(2i-1)(2j)}}, \nonumber\\
    Z_i \to & \prod_{j=1}^n X_j^{S_{(2i)(2j-1)}} Z_j^{S_{(2i)(2j)}}, \quad 1 \le i \le n,
\end{align}
which is actually another way to define the Clifford group \cite{Dehaene2003Cliffordgroup}.
With this row-vector convention\cite{Dehaene2003Cliffordgroup,Hostens2005Stabilizer}, symplectic matrices act on Pauli labels from the right. Consequently, if $C=C_1C_2$, then the corresponding phase-free symplectic representation is
\begin{equation}
S_C = S_{C_2}S_{C_1},
\end{equation}
i.e., the matrix product appears in the reverse order to the Clifford operator product.

Note that our choice of encoding and the definition of $\Omega$ differ slightly from common conventions in the literature\cite{Dehaene2003Cliffordgroup,aaronson2004improved}; the corresponding  matrices under the two conventions are related by simple permutations of rows and columns.

\subsection{Clifford disentanglers derived from MPS}

We construct a composite Clifford disentangler by sequentially applying two-qubit (2Q) or four-qubit (4Q) Clifford gates to an MPS obtained from a DMRG calculation, see Figure~\ref{fig: Clifford_disentangler}(a), with the objective of reducing the bipartite entanglement entropy across each MPS bond. The optimal gate is selected by minimizing the bipartite
1/2-R\'{e}nyi entropy
\begin{equation}
    S_\frac{1}{2}(\rho_A)=2\log(\Tr(\rho_A^{\frac{1}{2}}))=2\log\!\left(\sum_i \lambda_i\right),
\end{equation}
where $\rho_A$ is the reduced density matrix of subsystem $A$ and $\{\lambda_i\}$ are the Schmidt coefficients associated with the corresponding bipartition of the MPS. For simplicity, we denote $S_{1/2}$ as $S$. 
After optimization, the Clifford operator is applied to the Hamiltonian,
yielding a transformed Hamiltonian whose ground state has reduced
entanglement and can be represented better with the same bond dimension.
This transformed Hamiltonian can then be used in DMRG or VQE.
Figure \ref{fig: Clifford_disentangler}(b) illustrates the present variant of CAMPS algorithm, which differs from the previous algorithms\cite{qian_augmenting_2024}:
Clifford disentangling step is performed as a separate step
after a DMRG optimization, and the resulting MPS and transformed MPO initialize the subsequent DMRG calculation.
This separation allows us to construct MPO for the transformed Hamiltonian using the bipartite graph algorithm implemented in the \textsc{Renormalizer} package\cite{ren2020general}.

\subsection{2Q and 4Q Clifford disentanglers}
The number of distinct n-qubit Clifford operators is\cite{Koenig2014Cliffordenumerate}
\begin{align}
    N = 2^{n^2} \prod_{j=1}^n (4^j - 1),
\end{align}
which grows rapidly with $n$. 
The computational cost of the above disentangling step scales as $O(K N M^3)$. In this work, we consider 2Q Clifford gates in DMRG calculations with spin-orbitals, while 4Q Clifford gates are used in DMRG calculations with spatial orbitals.
For 2Q and 4Q Clifford operators, these numbers are $720$ and $47377612800$, respectively. Since our disentangling procedure requires searching over all possible local Clifford gates, direct enumeration quickly becomes computationally prohibitive, especially for 4Q disentanglers.

Since we are only concerned with the bipartite entanglement across a specific bond in the disentangling step,
we can introduce an equivalence relation among Clifford operators: two Clifford operators are said to belong to the same class if for any input MPS, they produce the same entanglement spectrum across the bond of interest. Operators within the same class differ only by two $\frac{n}{2}$-qubit Clifford operators acting on the first and the last $\frac{n}{2}$ qubits, as illustrated in Figure~\ref{fig: schematic-classification}. Such local transformations do not change the Schmidt spectrum and therefore do not affect the entanglement relevant to the optimization. Thus, the search for optimal local Clifford disentanglers can be confined to representatives of equivalence classes, resulting in orders-of-magnitude savings in computational cost.

\begin{figure}[!tbp]
	\begin{center}
		\includegraphics[width=\linewidth]{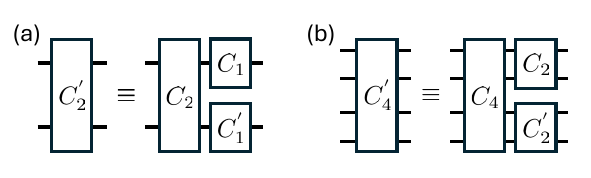}
	\end{center}
	\caption{
Schematics of 2Q and 4Q Clifford operators in the same equivalence class.
(a) Two 2Q Clifford operators $C_2$ and $C_2'$ belong to the same class if they differ only by arbitrary 1Q Clifford operators $C_1$ and $C_1'$ acting on the two qubits.
(b) Two 4Q Clifford operators $C_4$ and $C_4'$ belong to the same class if they differ only by arbitrary 2Q Clifford operators $C_2$ and $C_2'$ acting on qubits $(1,2)$ and $(3,4)$, respectively.}
\label{fig: schematic-classification}
\end{figure}

To identify these equivalence classes efficiently, we use the symplectic representation introduced before. Two Clifford operators with symplectic matrices $S_1$ and $S_2$ belong to the same class if
\begin{align}
    F \equiv S_2^{-1} S_1
    = \Omega S_2^T \Omega S_1
    =
    \begin{pmatrix}
        s & 0 \\
        0 & s'
    \end{pmatrix}
    \quad (\mathrm{mod}\;2),
\end{align}
where $s$ and $s'$ are $n \times n$ symplectic matrices. Equivalently, $S_2^{-1} S_1$ must be block diagonal, corresponding to independent Clifford transformations on the two halves. A brute-force classification based on pairwise comparison of all Clifford operators is still too expensive for large $N$, e.g, $N=4$. 
Therefore, we introduce a hash-based classification scheme, in which two Clifford operators are assigned to the same class if and only if they have the same hash value.  The classification algorithm and the proof of the equivalence relation are given in Appendix. Using this procedure, the full sets of 2Q and 4Q Clifford operators are reduced to $20$ and $91392$ distinct classes, respectively. This reduction substantially lowers the cost of the disentangling procedure and makes the use of 2Q and 4Q Clifford disentanglers practical in our calculations.


\subsection{Implementation details}
We implemented the Clifford disentangling procedure and the CAMPS algorithm in the \textsc{Camps} module\cite{campsmodule} of \textsc{Focus}.\cite{li2021expressibility,xiang_distributed_2024,li2025entanglement}
The electronic integrals were generated using PySCF.\cite{sun_recent_2020} 
The resulting second-quantized Hamiltonians were then transformed into qubit representations using the Jordan--Wigner\cite{jordan1928pauli} and parity\cite{whitfield2011simulation,seeley2012bravyi} mappings as implemented in OpenFermion.\cite{mcclean2020openfermion} 
For DMRG, the qubit Hamiltonians were converted automatically to MPO form using the bipartite graph algorithm implemented in \textsc{Renormalizer}\cite{ren2020general}. 
Since the Clifford disentanglers in general do not preserve the particle number symmetry, dense matrix algebra is used in CAMPS.
For the 2Q and 4Q classification, Clifford operators were generated with the \textsc{QuantumClifford} Julia package.\cite{QuantumClifford_github}

\begin{figure*}[!tbp]
    \centering
    \includegraphics[width=0.96\textwidth]{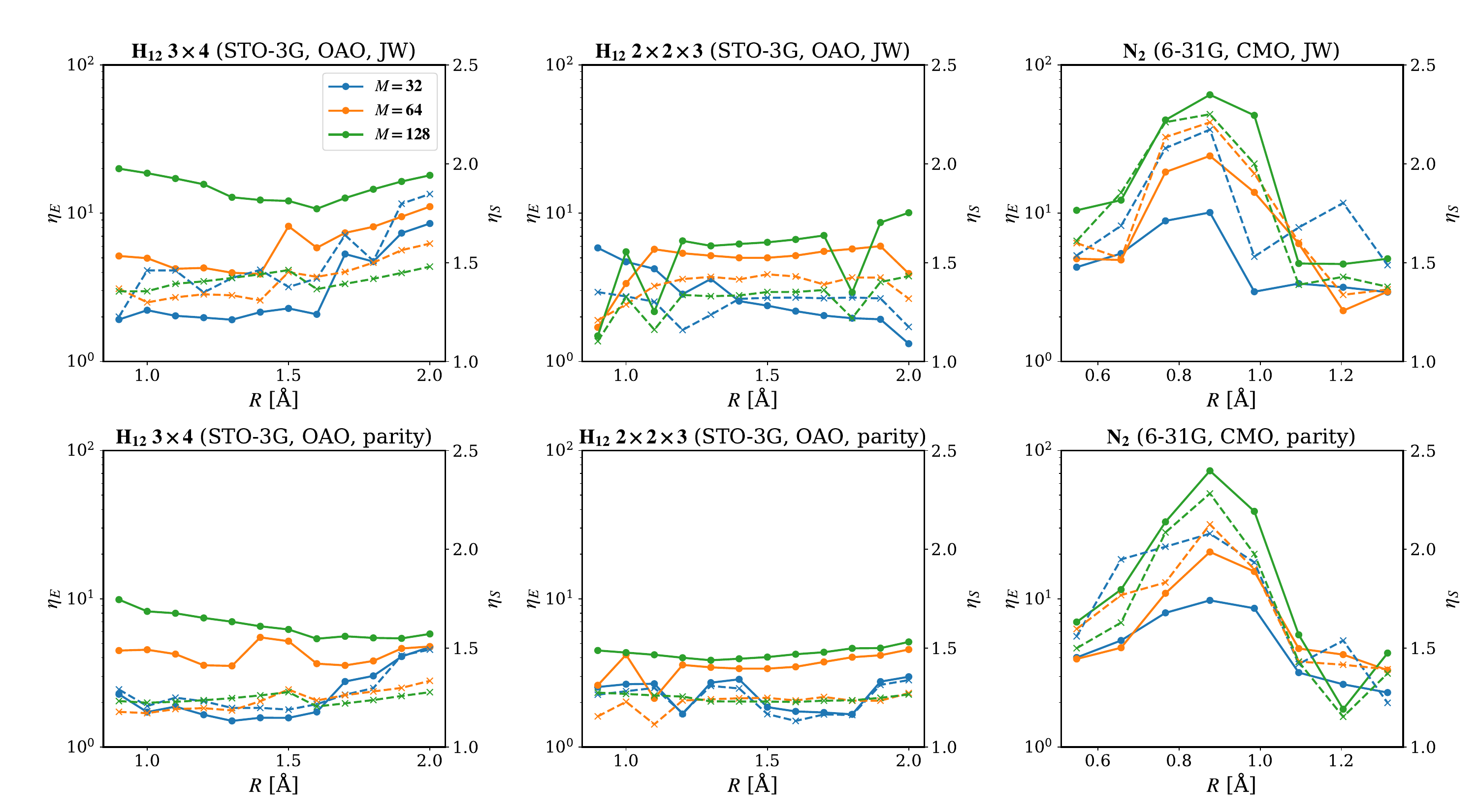}
    \caption{Bond-length dependence of the relative energy error ratio $\eta_E$ (solid lines) between DMRG and CAMPS and the ratio of the maximum entanglement entropies $\eta_S$ (dashed lines) for the MPS obtained from conventional DMRG and after Clifford disentangling step. Upper: JW mapping, Lower: parity mapping.
    }
    \label{fig: ratio_bonds}
\end{figure*}

\begin{figure*}[!t]
    \centering
    \includegraphics[width=0.96\textwidth]{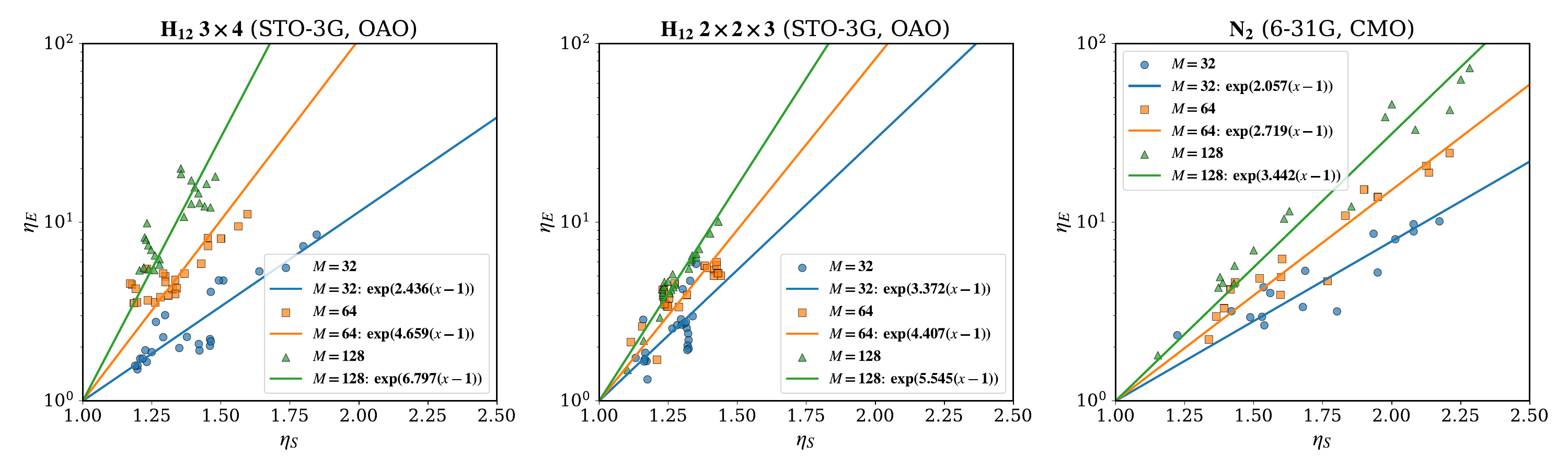}
    \caption{
    Correlation between the relative energy error ratio $\eta_E$ obtained from DMRG and CAMPS and the ratio of the maximum entanglement entropies $\eta_S$ for the MPS before and after Clifford disentangling step. Data for both JW and parity mappings in Figure \ref{fig: ratio_bonds} are used
    and fitted using $\eta_E=\exp[k(\eta_S-1)]$.     
    }
    \label{fig: ratio_scaling}
\end{figure*}

\section{Results and discussion}
\subsection{Effectiveness of Clifford disentanglers}
In this section, we assess the effectiveness of Clifford disentanglers. 
We characterized the performance from two complementary perspectives. First, we examine the change in the bipartite entanglement entropy across bond before and after disentangling, denoted as $S_{m,m+1}^{\mathrm{MPS}}$ and $S_{m,m+1}^{\mathrm{disen}}$, respectively, where the entropy is evaluated using the $\tfrac{1}{2}$-Rényi form. Second, we compare the ground-state energies obtained from DMRG calculations for the original and Clifford-transformed Hamiltonians, denoted by $E_{\mathrm{MPS}}$ and $E_{\mathrm{CAMPS}}$, respectively.
To facilitate a more quantitative comparison, we further employ two dimensionless measures: the relative energy error ratio $\eta_E$,\cite{huang_clifford_2025,huang_nonstabilizerness_2025} which characterizes the improvement in the ground-state energy, and the maximum-entropy ratio $\eta_S$, which characterizes the reduction in bipartite entanglement induced by Clifford disentangling. They are defined as
\begin{align}
    \eta_E &= \frac{E_{\mathrm{MPS}} - E_0}{E_{\mathrm{CAMPS}} - E_0}, \\
    \eta_S &= \frac{S_{\max}^{\mathrm{MPS}}}{S_{\max}^{\mathrm{disen}}},
\end{align}
where $E_0$ is the ground-state energy obtained from full configuration interaction (FCI), $E_{\mathrm{MPS}}$ and $E_{\mathrm{CAMPS}}$ are the DMRG ground-state energies before and after the Clifford transformation, respectively, and $S_{\max}^{\mathrm{MPS}}$ and $S_{\max}^{\mathrm{disen}}$ \cite{gxdn-zwrw} are the corresponding maximum values over all bonds.

We illustrate the above analysis using three representative systems: \ce{H12} in the STO-3G basis\cite{hehre1969self} with orthonormalized atomic orbitals (OAOs) in both the planar $3 \times 4$ and cubic $2\times 2 \times 3$ geometries, and \ce{N2} in the 6-31G basis\cite{hehre1972self} with canonical molecular orbitals (CMOs).
Figure \ref{fig: ratio_bonds} illustrates the dependence of $\eta_{E}$ (solid lines) and $\eta_S$ (dashed lines) on the bond length for different molecules under JW and parity mappings. The orbital ordering for each system is held fixed across bond lengths 
to reveal the trends solely due to the effects of Clifford disentanglers.
Under these conditions, we observe that large energy-error reductions are associated with larger entropy reductions. Thus, evaluating $\eta_S$ after Clifford disentangling
step provides a low-cost qualitative diagnostic for whether a subsequent
CAMPS optimization is likely to be beneficial.

Figure \ref{fig: ratio_scaling} summarizes the relationship between $\eta_E$ and $\eta_S$ with all the data in Figure \ref{fig: ratio_bonds}.
A clear linear correlation is observed between $\ln(\eta_E)$ and $\eta_S$, and the data can be well fitted by $\eta_E=\exp[k(\eta_S-1)]$. This behavior indicates that a reduction in the maximum entanglement entropy leads to an exponential improvement in the energy obtained from DMRG with a Clifford-transformed Hamiltonian. 
The slope $k$ seems to increase monotonically with the bond dimension $M$. 
It is worth mentioning that the slope $k$ also depends on other factors, such as the specific molecular systems and the orbital orderings. Nevertheless, 
this linear relationship supports the use of the entropy reduction as a practical
screening diagnostic before performing the more expensive CAMPS energy optimization.

\begin{figure}[!b]
    \centering
    \includegraphics[width=\columnwidth]{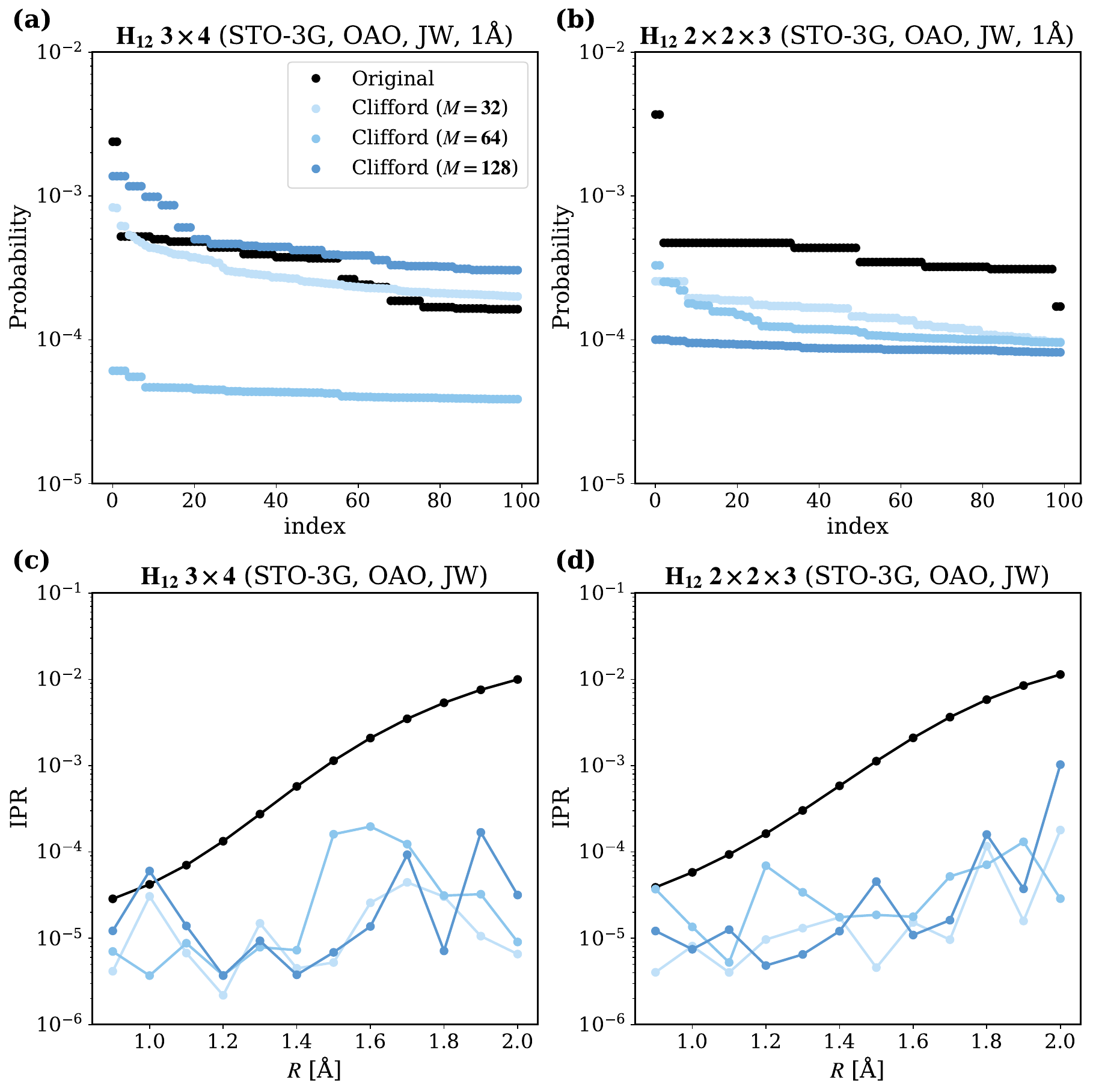}
    \caption{
Spreadness of the ground-state wavefunction for the original
Hamiltonian and the Clifford-transformed Hamiltonian.
    (a,b) Probability distributions of the exact ground-state wavefunction for the $\mathrm{H}_{12}$ system in the STO-3G basis (OAO, Jordan--Wigner mapping, bond length 1 {\AA}). Panels (a) and (b) correspond to the $3\times4$ and $2\times2\times3$ geometries, respectively. The probabilities of the 100 most significant configurations are plotted, with the Clifford disentanglers obtained from DMRG calculations using three values of $M$ (32, 64, and 128) for comparison.
 (c,d) Inverse participation ratio (IPR) of the exact ground state as a function of the bond length for the $3\times4$ and $2\times2\times3$ geometries, respectively.}
    \label{fig: IPR}
\end{figure}

\begin{figure*}[!t]
    \centering
    \includegraphics[width=0.96\textwidth]{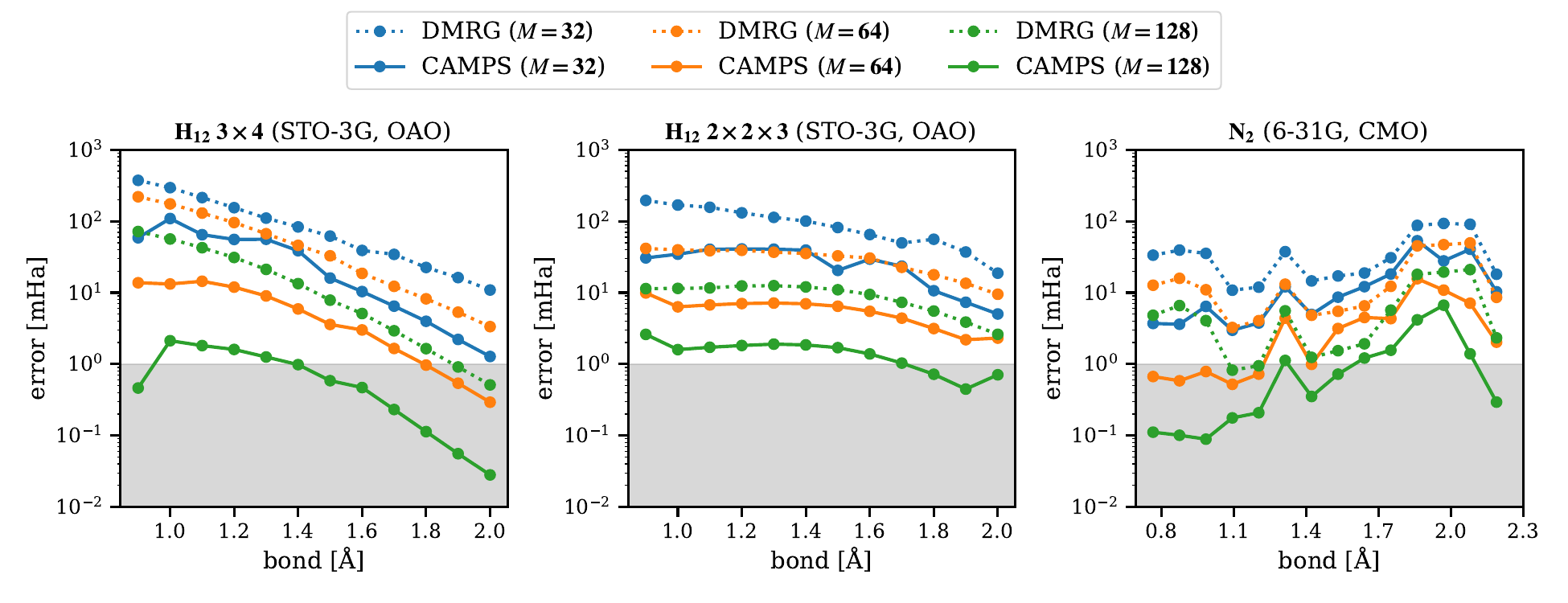}
    \caption{Energy errors relative to the FCI reference obtained from DMRG and CAMPS calculations with different bond dimensions $M$ using the Jordan--Wigner (JW) mapping. 
    }
    \label{fig: energy_error}
\end{figure*}

Finally, we mention that contrary to Givens rotations\cite{li2025entanglement}, the Clifford disentanglers do not increase the concentration of configurations in the exact ground state of the transformed Hamiltonian; instead, the distribution becomes slightly more delocalized. As shown in Figures \ref{fig: IPR}(a) and (b), the ground-state configuration probabilities of the Clifford-transformed \ce{H12} systems (planar $3\times 4$ and cubic $2\times 2\times 3$) decay more slowly and appear more evenly distributed. This trend is further confirmed by the inverse participation ratio (IPR), viz., $\mathrm{IPR}=\sum_{s_1\cdots s_K} |\langle s_1\cdots s_K|\Psi\rangle|^4$, 
shown in Figures \ref{fig: IPR}(c) and (d), which quantifies the configuration concentration and indicates that the ground-state configurations of the transformed Hamiltonian are generally more delocalized. This behavior originates from the structure of Clifford gates $\{\mathrm{CNOT}, \mathrm{S}, \mathrm{H}\}$. While CNOT and S gates preserve the IPR, the Hadamard gate will redistribute the weight of dominant configurations, thereby producing a more uniform configuration distribution. 

\subsection{Clifford disentanglers for DMRG}
In this section, we investigate the effect of incorporating the Clifford disentangler into the DMRG algorithm. The CAMPS workflow is summarized in Figure \ref{fig: Clifford_disentangler}(b). 
The preceding analysis held orbital orderings fixed.
In practical DMRG calculations, however, orbital ordering strongly
affects accuracy.\cite{chan2011density}
We therefore use the Fiedler orderings\cite{barcza2011quantum, olivares2015ab} for the baseline calculations in this section and separately test sensitivity to perturbed
orderings.

Figure \ref{fig: energy_error} compares 2Q CAMPS with conventional DMRG
for \ce{H12} in OAOs and \ce{N2} in CMOs, using the JW mapping.
Across the plotted geometries, CAMPS lowers the energy error at a fixed
bond dimension. With $M=128$, most CAMPS points fall within the chemical accuracy (1.6 mHa). The errors of CAMPS are broadly comparable to those from conventional DMRG with
bond dimension $2M$. These results establish the numerical benefits of Clifford
disentanglers for the systems studied.


\newsavebox{\swapCircuit}
\newsavebox{\swapCircuitDecomposed}
\newsavebox{\fswapCircuit}
\newsavebox{\fswapCircuitDecomposed}

\begin{figure*}[!t]
\centering

\begin{tikzpicture}
    \node[anchor=north west, inner sep=0pt] (img) at (0,0)
        {\includegraphics[width=0.94\textwidth]{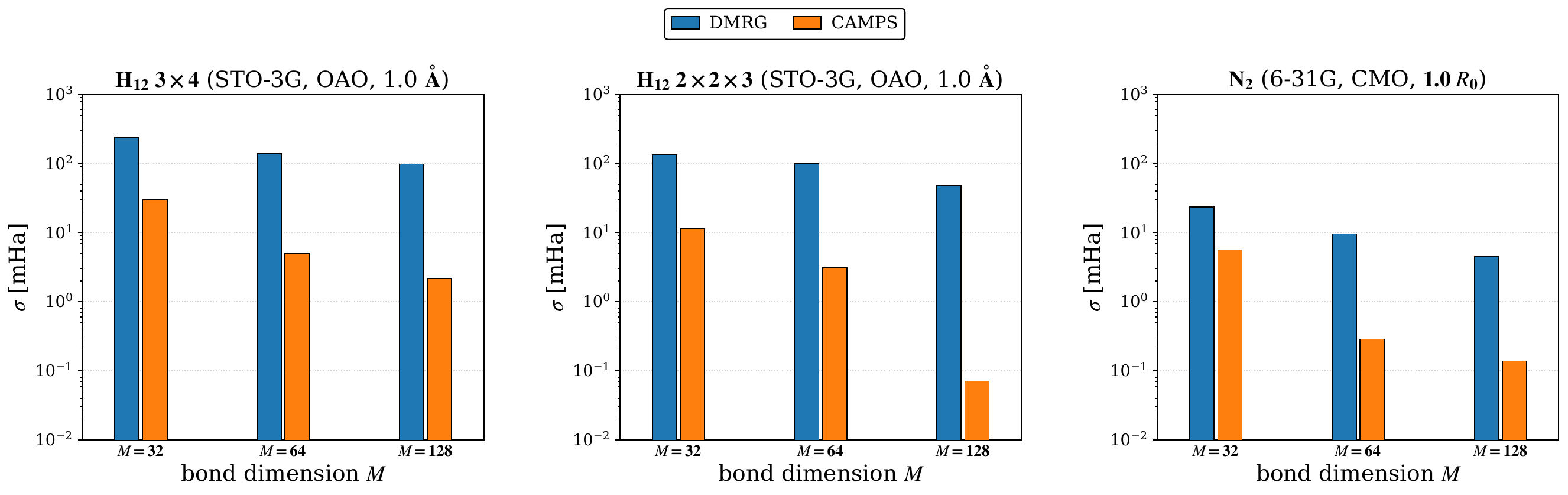}};
    \node[
        anchor=north west,
        inner sep=1pt,
        fill=white,
        fill opacity=0.8,
        text opacity=1
    ] at ([xshift=2mm,yshift=-2mm]img.north west)
        {(a)};
\end{tikzpicture}

\vspace{0.25cm}

\begin{lrbox}{\swapCircuit}
\begin{quantikz}[row sep=0.4cm, column sep=0.2cm]
    \lstick{$q_0$} & \swap{1} & \qw \\
    \lstick{$q_1$} & \targX{} & \qw
\end{quantikz}
\end{lrbox}

\begin{lrbox}{\swapCircuitDecomposed}
\begin{quantikz}[row sep=0.4cm, column sep=0.2cm]
    \lstick{$q_0$} & \ctrl{1} & \targ{}   & \ctrl{1} & \qw \\
    \lstick{$q_1$} & \targ{}  & \ctrl{-1} & \targ{}  & \qw
\end{quantikz}
\end{lrbox}

\begin{lrbox}{\fswapCircuit}
\begin{quantikz}[row sep=0.4cm, column sep=0.2cm]
    \lstick{$q_0$} & \gate[2]{\mathrm{FSWAP}} & \qw \\
    \lstick{$q_1$} &                            & \qw
\end{quantikz}
\end{lrbox}

\begin{lrbox}{\fswapCircuitDecomposed}
\begin{quantikz}[row sep=0.4cm, column sep=0.2cm]
    \lstick{$q_0$} & \swap{1} & \qw      & \ctrl{1} & \qw       & \qw \\
    \lstick{$q_1$} & \targX{} & \gate{H} & \targ{}  & \gate{H}  & \qw
\end{quantikz}
\end{lrbox}

\begin{tabular}{@{}c@{\hspace{1.6cm}}c@{}}

\begin{tabular}[t]{@{}l@{\hspace{0.35cm}}l@{}}
\raisebox{0.7em}{(b)} &
\usebox{\swapCircuit}
\hspace{0.2cm}=\hspace{0.2cm}
\usebox{\swapCircuitDecomposed}
\end{tabular}

&

\begin{tabular}[t]{@{}l@{\hspace{0.35cm}}l@{}}
\raisebox{0.7em}{(c)} &
\usebox{\fswapCircuit}
\hspace{0.2cm}=\hspace{0.2cm}
\usebox{\fswapCircuitDecomposed}
\end{tabular}

\end{tabular}

\caption{
Orbital ordering dependence of DMRG and CAMPS.
(a) Standard deviation of energy errors obtained with three orderings,
\(\sigma_X =
[\frac{1}{3}\sum_{i=1}^{3}
(\Delta E_X(i)-\overline{\Delta E_X})^2]^{1/2}\),
with \(X\in\{\mathrm{DMRG},\mathrm{CAMPS}\}\), using the Jordan--Wigner mapping,
relative to the FCI reference.
For \ce{N2}, the equilibrium bond length is \(R_0=1.095\)~\AA.
Three orbital orderings are considered: Fiedler ordering, a partial permutation
of the Fiedler ordering, and a random permutation.
(b,c) Circuit decompositions of the SWAP and fermionic SWAP (FSWAP) gates into
Clifford gates.
}
\label{fig: different_orders}
\end{figure*}

We next examine the sensitivity of CAMPS to orbital ordering.
In conventional DMRG, the ordering of orbitals along the MPS chain can
strongly affect accuracy because placing correlated orbitals close together
generally reduces the required bond dimension.
To assess the sensitivity of CAMPS to orbital ordering, we consider three representative orderings:
the Fiedler ordering obtained from the exchange integrals $K_{ij}=[ij|ji]$ (order 1),
a partial permutation of that ordering (order 2),
and a random permutation (order 3).
For method $X\in\{\mathrm{DMRG},\mathrm{CAMPS}\}$, we define $\Delta E_X(i)=E_X(i)-E_0$ and quantify ordering sensitivity by the standard deviation of the energy errors $\sigma_X=\left[\frac{1}{3}\sum_{i=1}^{3}\left(\Delta E_X(i)-\overline{\Delta E_X}\right)^2\right]^{1/2}$. Figure~\ref{fig: different_orders} shows that CAMPS gives smaller $\sigma_X$ than conventional DMRG for the tested \ce{H12} and \ce{N2} cases.
This reduced sensitivity can be rationalized by the structure of the Clifford disentanglers. The set of two-qubit Clifford gates contains both SWAP and fermionic SWAP (FSWAP) operations (see Figure \ref{fig: different_orders}),
allowing the optimization to partially rearrange orbitals while reducing entanglement. Consequently, CAMPS is less sensitive to orbital orderings.

\newsavebox{\jwParityCircuitA}
\newsavebox{\jwParityCircuitB}
\newsavebox{\jwParityCircuitBDecomposed}

\begin{figure*}[!t]
    \centering

\begin{lrbox}{\jwParityCircuitA}
\begin{quantikz}[row sep=0.4cm, column sep=0.2cm]
    \lstick{$q_0$} & \ctrl{1} & \qw      & \qw      & \qw \\
    \lstick{$q_1$} & \targ{}  & \ctrl{1} & \qw      & \qw \\
    \lstick{$q_2$} & \qw      & \targ{}  & \ctrl{1} & \qw \\
    \lstick{$q_3$} & \qw      & \qw      & \targ{}  & \qw
\end{quantikz}
\end{lrbox}

\begin{lrbox}{\jwParityCircuitB}
\begin{quantikz}[row sep=0.4cm, column sep=0.2cm]
    \lstick{$q_0$} & \ctrl{1} & \qw      & \qw      & \qw \\
    \lstick{$q_1$} & \targ{}  & \ctrl{2} & \qw      & \qw \\
    \lstick{$q_2$} & \qw      & \qw      & \ctrl{1} & \qw \\
    \lstick{$q_3$} & \qw      & \targ{}  & \targ{}  & \qw
\end{quantikz}
\end{lrbox}

\begin{lrbox}{\jwParityCircuitBDecomposed}
\begin{quantikz}[row sep=0.3cm, column sep=0.2cm]
    \lstick{$q_0$} & \ctrl{1} & \qw      & \qw      & \qw      & \qw      & \qw \\
    \lstick{$q_1$} & \targ{}  & \swap{1} & \qw      & \swap{1} & \qw      & \qw \\
    \lstick{$q_2$} & \qw      & \targX{} & \ctrl{1} & \targX{} & \ctrl{1} & \qw \\
    \lstick{$q_3$} & \qw      & \qw      & \targ{}  & \qw      & \targ{}  & \qw
\end{quantikz}
\end{lrbox}

\begin{tabular}{c}

\begin{tikzpicture}
    \node[anchor=north west, inner sep=0pt] (img) at (0,0)
        {\includegraphics[width=0.92\textwidth]{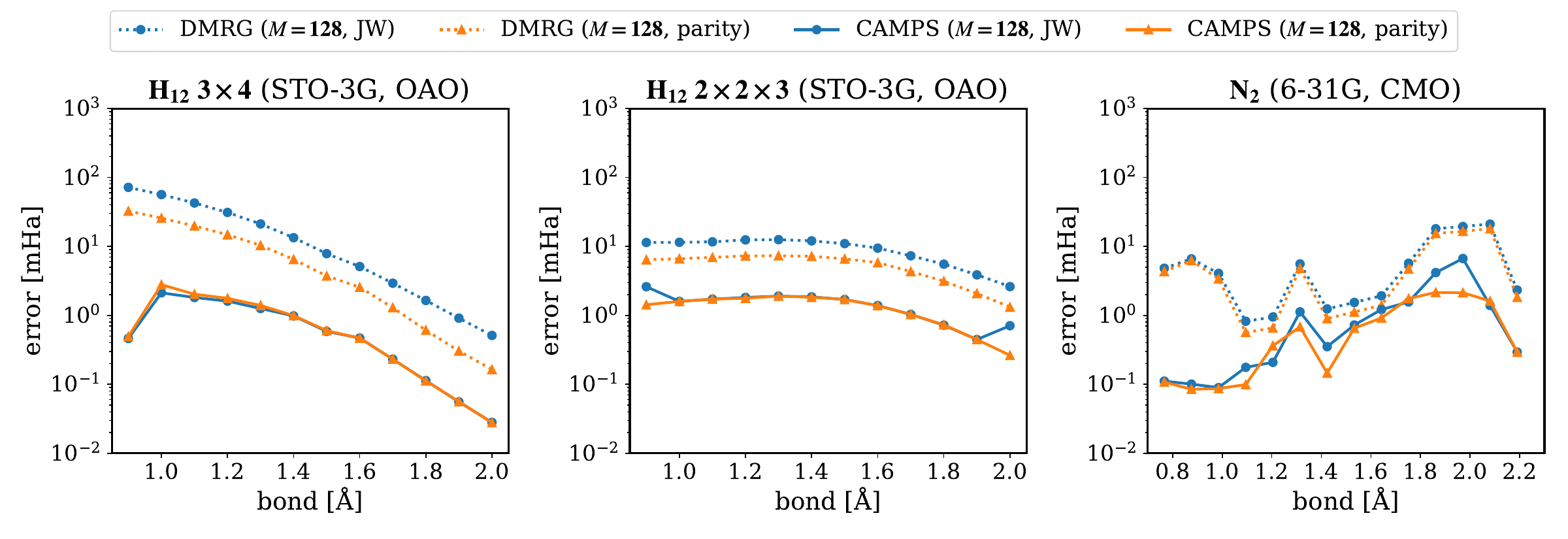}};
    \node[
        anchor=north west,
        inner sep=1pt,
        fill=white,
        fill opacity=0.8,
        text opacity=1
    ] at ([xshift=2mm,yshift=-2mm]img.north west)
        {(a)};
\end{tikzpicture}
\\[0.2cm]

\begin{tabular}{@{}c@{\hspace{1.6cm}}c@{}}

\begin{tabular}[t]{@{}l@{\hspace{0.35cm}}l@{}}
\raisebox{1.8em}{(b)} &
\usebox{\jwParityCircuitA}
\end{tabular}

&

\begin{tabular}[t]{@{}l@{\hspace{0.35cm}}l@{}}
\raisebox{1.8em}{(c)} &
\usebox{\jwParityCircuitB}
\hspace{0.2cm}=\hspace{0.2cm}
\usebox{\jwParityCircuitBDecomposed}
\end{tabular}

\end{tabular}
\end{tabular}

    \caption{
    Fermion-to-qubit mapping dependence of DMRG and CAMPS.
    (a) Energy errors with respect to full configuration interaction (FCI) as functions of bond length obtained from DMRG and CAMPS calculations with a fixed bond dimension $M=128$ under two different fermion-to-qubit mappings (Jordan--Wigner and parity).
    (b,c) Circuit representations of basis transformations.
    (b) Jordan--Wigner to parity encoding.
    (c) Jordan--Wigner to Bravyi--Kitaev encoding.
    }
    \label{fig: diff_mappings}
\end{figure*}

\begin{figure*}[!t]
    \centering
    \includegraphics[width=0.90\textwidth]{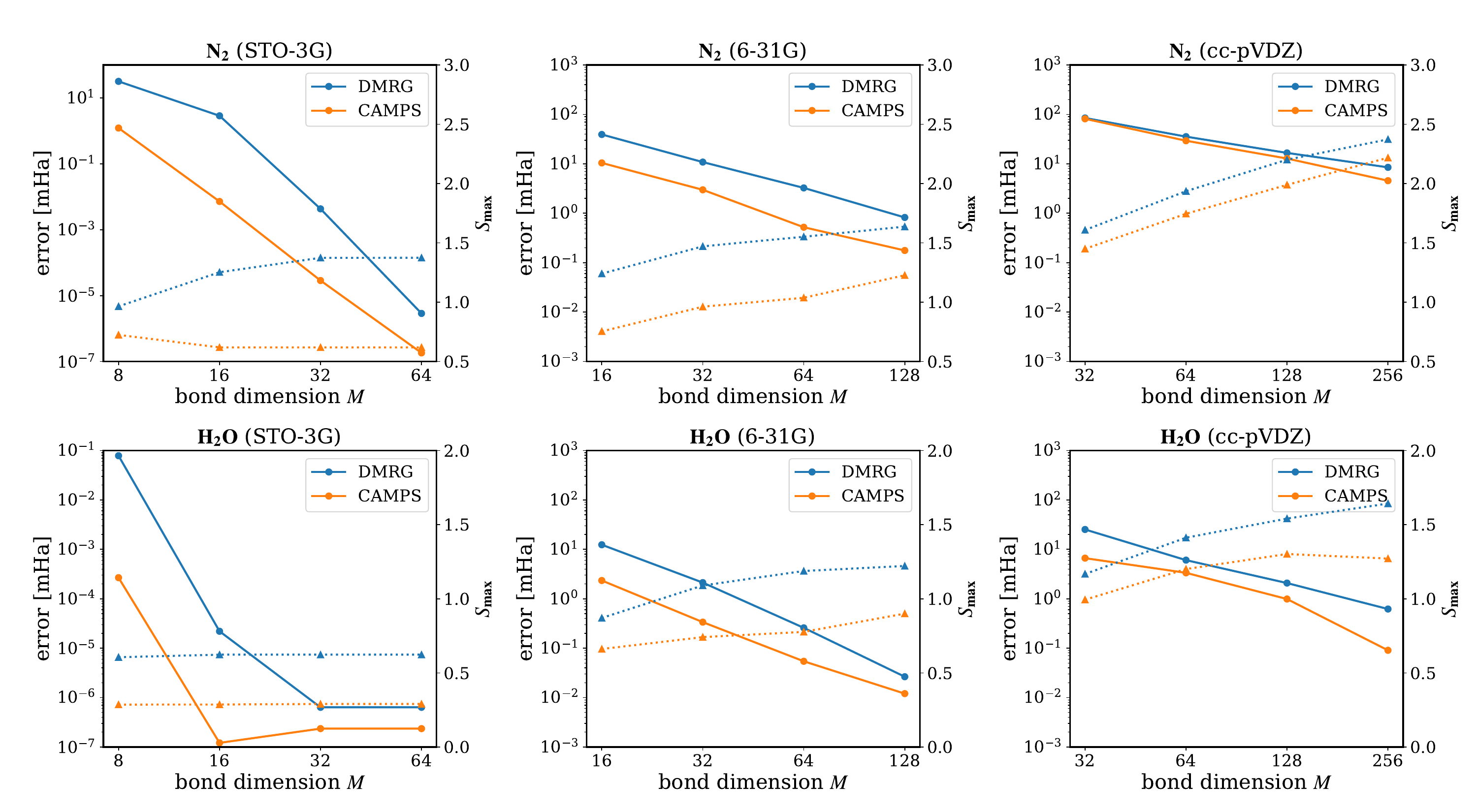}
    \caption{
	Basis set dependence of DMRG and CAMPS.    
    Energy errors (solid lines) and maximal bipartite $1/2$-R\'enyi entropies (dashed lines) over all MPS bonds as functions of bond dimension $M$ for \ce{N2} and \ce{H2O} at their equilibrium geometries. The reference energies were obtained from FCI calculations for all cases except \ce{N2} (cc-pVDZ), for which a nearly exact spin-adapted DMRG result\cite{wang2025generalized} with $M=5000$ was used.}
    \label{fig: different_basis}
\end{figure*}

At fixed orbital ordering, we next examine the impact of the Jordan--Wigner and parity mappings on CAMPS (Figure~\ref{fig: diff_mappings}). For \ce{H12} in the OAO representation, CAMPS results from the two mappings are very similar. For \ce{N2} in the CMO representation, a larger mapping dependence is observed with local 2Q disentanglers. Since transformations between Jordan--Wigner, parity, and Bravyi--Kitaev encodings can be expressed as Clifford circuits as illustrated for four qubits in Figure~\ref{fig: diff_mappings}, we can expect that, if $n$Q Clifford gates are employed for disentangling instead of local 2Q Clifford gates, quantum states obtained using these mapping can converge to the same CAMPS solution.
Therefore, the observed dependence on mappings only reflects the limitation of the local search.

\begin{figure*}[!t]
    \centering
    \includegraphics[width=0.80\textwidth]{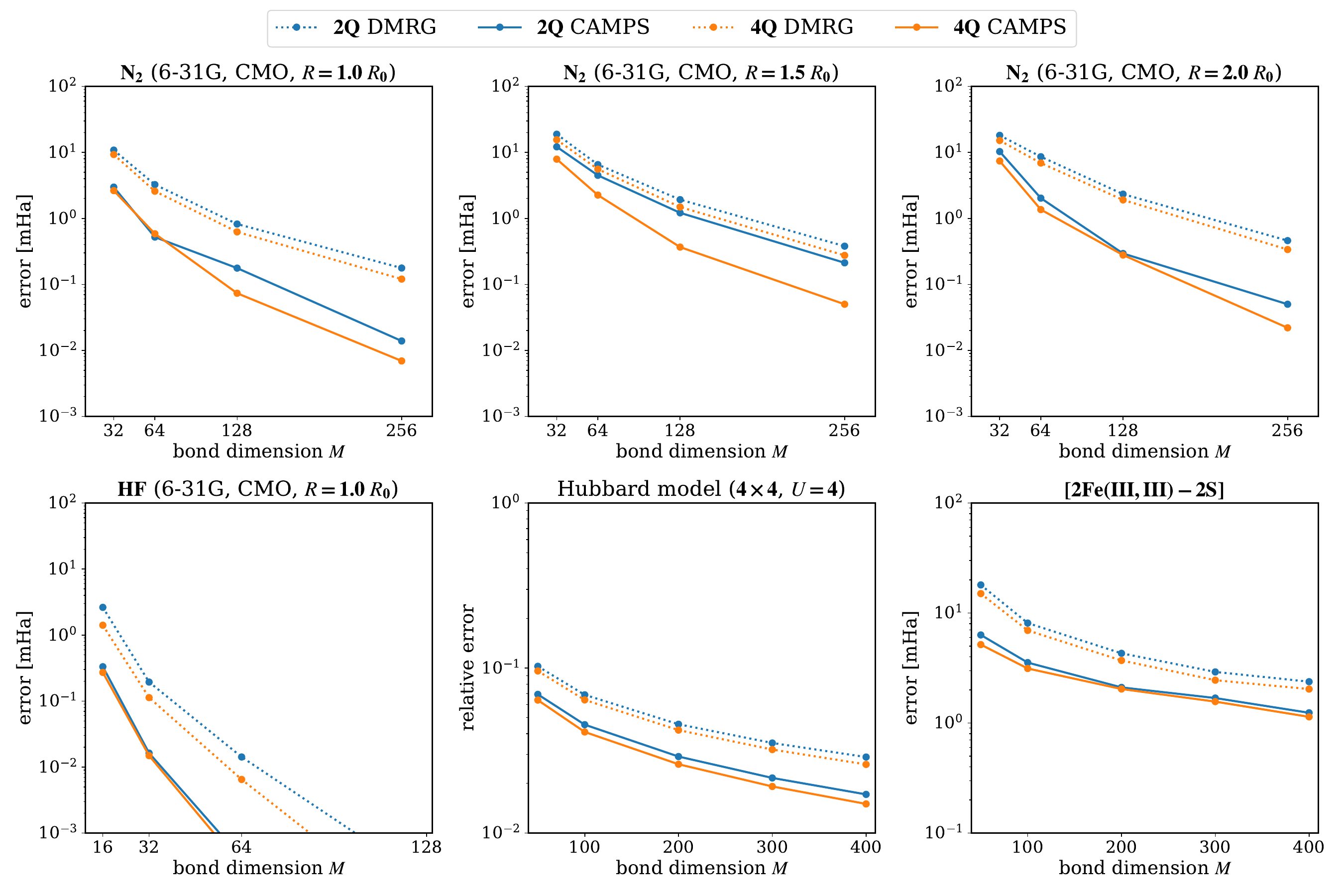}
    \caption{
	Comparison between 2Q and 4Q DMRG and CAMPS. 	   
    Energy errors of DMRG and CAMPS as a function of bond dimension $M$ 
    for the ground-state energies of various systems, including \ce{N2} at different bond lengths (1.0$R_0$, 1.5$R_0$, and 2.0$R_0$ with $R_0=1.095$\AA), \ce{HF} at its equilibrium bond length $R_0=0.9168$\AA, 
    the half-filled $4\times4$ Hubbard model ($U=4$), and the [{2Fe(III,III)-2S}] cluster. The reference ground-state energies are taken from full configuration interaction (FCI) calculations. 2Q CAMPS denotes calculations performed in the spin-orbital basis with 2Q Clifford disentanglers, whereas 4Q CAMPS denotes calculations in the spatial-orbital basis with 4Q Clifford disentanglers.
    }
    \label{fig: 2q_vs_4q}
\end{figure*}

We further examine basis-set dependence in Figure~\ref{fig: different_basis}. For both \ce{N2} and \ce{H2O}, the improvement of CAMPS relative to conventional DMRG decreases as the basis is enlarged in the calculations shown. 
The rate of this reduction depends strongly on the character of correlation in the system. The reduction is modest for \ce{H2O} from 6-31G to cc-pVDZ, indicating that the additional entanglement introduced by the larger basis set does not substantially alter the entanglement structure relevant to Clifford disentangling for \ce{H2O}. In contrast, enlarging the basis set for \ce{N2} introduces a greater amount of dynamic correlation, which creates a larger amount of entanglement  that is more difficult to remove using Clifford transformations. This is reflected by the smaller ratio between $S_{\max}$ before and after Clifford
disentangling step, see Figure~\ref{fig: different_basis}.
    
Finally, we compare the effects of 2Q and 4Q Clifford disentanglers for CAMPS. 
The 2Q DMRG/CAMPS calculations use a spin-orbital representation, whereas the 4Q DMRG/CAMPS calculations use a spatial-orbital MPS. 
As shown in Figure~\ref{fig: 2q_vs_4q}, 
except for \ce{N2} at $1.5R_0$ with $R_0=1.095$\AA,
4Q disentanglers offer no substantial improvement over 2Q disentanglers for the displayed \ce{N2}, \ce{HF}, half-filled 4-by-4 Hubbard-model ($U=4$), and the CAS(30e,20o) active space model\cite{li2017spin,linkToFCIDUMPfe2fe4} of a [2Fe(III,III)--2S] cluster, while the computational cost increases by several orders of magnitude.
This behavior can be rationalized from the structure of the Clifford group, where 4Q Clifford operators can be formed from products of 2Q Clifford gates.
As a result, the possible advantage of the 4Q scheme seems to be limited, 
as 2Q formulation benefits from more local disentangling operations 
and a substantially lower computational cost.

\begin{figure*}[!tbp]
    \centering
    \includegraphics[width=0.90\textwidth]{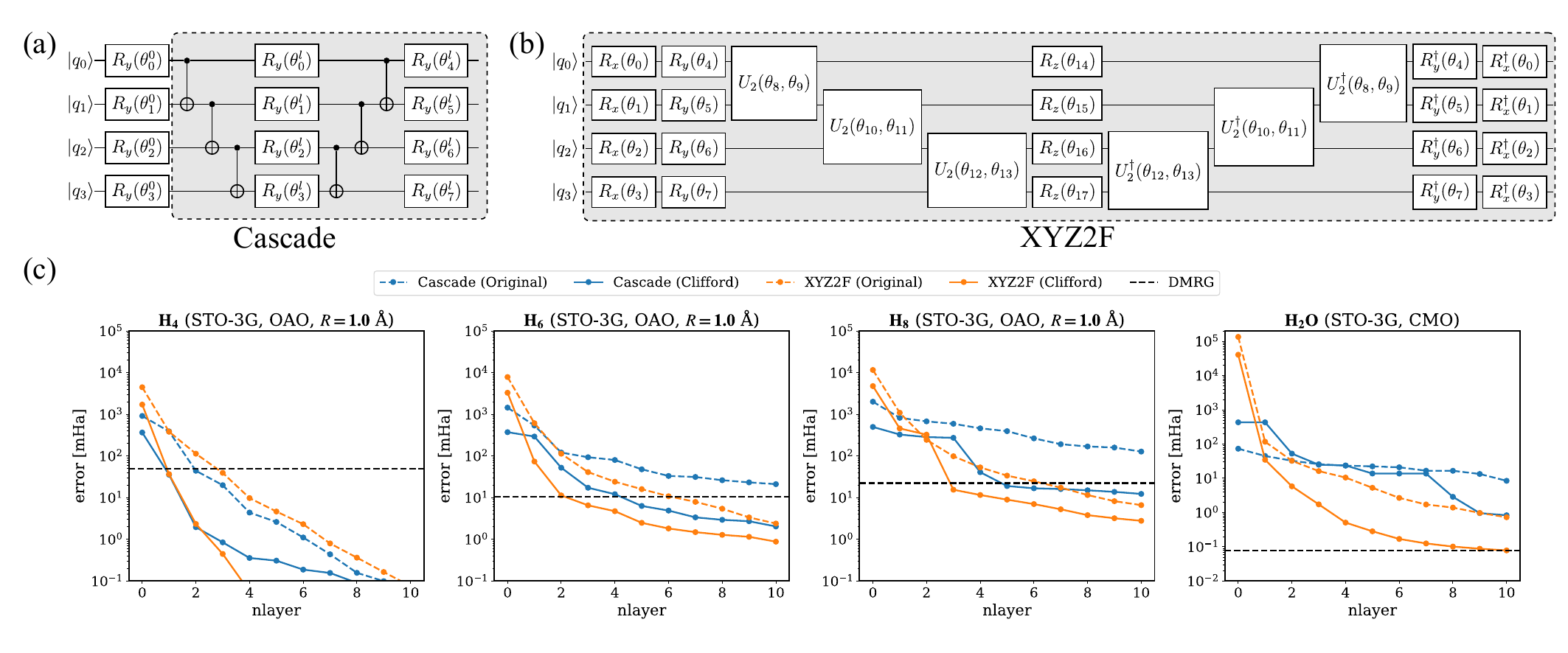}
    \caption{
	Clifford disentanglers for VQE.    
    (a) Cascade ansatz; 
    (b) XYZ2F ansatz;
    (c) Energy errors obtained from VQE as functions of the number of layers in the Cascade and XYZ2F circuits. 
    For \(\mathrm{H}_4\), the Clifford transformed Hamiltonian is generated using an MPS with bond dimension $M=4$, while for the remaining systems $M=8$ is used. The energies obtained from the corresponding MPS are also illustrated by black dashed lines.}
    \label{fig: VQE}
\end{figure*}

\subsection{Clifford disentanglers for VQE}
To illustrate the usefulness of Clifford disentanglers
for quantum simulations,
we also perform VQE calculations using both the original qubit Hamiltonian 
with the Jordan--Wigner mapping and the Clifford-transformed
Hamiltonian obtained with Clifford disentanglers determined from a low-bond-dimension MPS. 
The Cascade\cite{CASCADE} and XYZ2F\cite{XYZ2F} circuits are employed as variational ans\"atze with $\ket{00\cdots 0}$ as the initial state. The circuits are shown in Figure~\ref{fig: VQE}, and a layerwise optimization strategy\cite{XYZ2F} is adopted.

As shown in Figure \ref{fig: VQE}, VQE with the Clifford-transformed Hamiltonian yields substantially more accurate results than VQE with the original Hamiltonian. 
Moreover, at relatively small circuit depths, the VQE energies obtained
with the Clifford-transformed Hamiltonian
already surpass those obtained from the MPS used to generate the corresponding Hamiltonian. These tests show that Clifford disentanglers can provide useful Hamiltonian preprocessing for quantum simulations. 
By transferring part of the ground-state entanglement into a classically optimized Clifford transformation, the resulting ground state of the Clifford-transformed Hamiltonian can be represented more efficiently by shallower circuits.


\section{Conclusion}
We systematically assess the use of Clifford disentanglers for reducing entanglement 
in electronic structure simulations. In particular, the use of
low-bond-dimension MPS to determine Clifford disentanglers
and the subsequent application of Clifford disentanglers
to CAMPS and VQE are studied. Compared to ordinary DMRG, we find the CAMPS algorithm exhibits several useful features:

\begin{enumerate}
    \item \textit{Improved accuracy.}
    For the systems tested, CAMPS reduces the energy error at fixed bond dimension and can achieve an accuracy comparable to conventional DMRG calculations with a larger bond dimension.

    \item \textit{Entropy-based diagnostic.}
    The reduction of the maximum bipartite $1/2$-R\'enyi entropy is correlated with the improved energies, providing a simple empirical criterion for estimating the benefit of Clifford disentanglers.

    \item \textit{Reduced sensitivity to orbital ordering and fermion-to-qubit mappings.}
    CAMPS is less sensitive than conventional DMRG to orbital ordering and can partially mitigate the dependence on the fermion-to-qubit mapping, as swaps and mappings can all be expressed as Clifford circuits.

    \item \textit{Correlation-dependent performance.}
    The observed benefit of Clifford disentanglers decreases as the basis-set size increases, especially for \ce{N2}. 
    Presumably, dynamical correlation introduces entanglement that cannot be effectively
    removed by Clifford circuits. Further investigations on more systems are required to generalize this trend.

\end{enumerate}

An important contribution of this work is the efficient classification of Clifford operators based on the symplectic representation.
By grouping Clifford operators that are equivalent with respect to the entanglement spectrum across a bipartition, the search space is reduced from 720 to 20 classes for 2Q Clifford operators and from 47377612800 to 91392 classes for 4Q Clifford operators.
This reduction makes the search for practical Clifford disentanglers computationally feasible.
Preliminary numerical tests show that in most cases, 4Q Clifford disentanglers did not yield significant additional improvement over 2Q Clifford disentanglers.
However, the classification is still valuable for future investigation with
4Q Clifford disentanglers. For instance, one can sample a few representatives
or use machine learning to predict important representatives, which can
significantly lower the cost of using 4Q Clifford disentanglers.

The computational scaling of CAMPS is formally the same as DMRG. However, due to the loss of U(1) symmetry, it is more expensive in practice. Future work needs to investigate the possibility of restoring U(1) symmetry, e.g., by only considering
U(1) preserving Clifford disentanglers.

Apart from CAMPS, we also demonstrate that the use of Clifford-transformed Hamiltonians can also improve VQE calculations by transferring part of the many-body entanglement into a classically optimized Clifford transformation. In summary, the results in the present work suggest that Clifford disentanglers, complementary to Givens rotations,\cite{li2025entanglement,huang2026augmenting} can serve as useful structure-preserving preprocessing tools for both tensor-network and quantum simulation of electronic structure problems. Combining these transformations with tensor-network algorithms and shallow quantum circuits may provide a practical route toward lower-complexity simulations of strongly correlated molecular systems.

\FloatBarrier

\appendix
\setcounter{equation}{0}
\renewcommand{\theequation}{A\arabic{equation}}

\section*{Appendix}

\subsection*{Hash-Based Classification of Clifford Operators}

This section describes an efficient hash-based classification algorithm for $n$-qubit Clifford operators (see Algorithm \ref{algo: hashing_algo}) for $n=2$ or 4.
Because the binary symplectic representation retains only the phase-free
Pauli action, the classification is performed for Clifford operators modulo
Pauli factors and global phases, rather than for all Clifford unitaries.

Then, we show that two Clifford operators belong to the same equivalence class if and only if the corresponding symplectic matrices $T_1$, $T_2$ and $T_3$ are identical. 
Therefore, guided by the hash value, we only need to compare each operator with operators that have the same hash value. The complete classification procedure is shown in Algorithm \ref{algo: classification}.
\begin{algorithm}[H]
\caption{Hashing the symplectic matrix of a 2-qubit or 4-qubit Clifford operator to an integer}
\begin{algorithmic}[1]

\Require 
    \Statex \quad a binary $2n\times 2n$ symplectic matrix $S$, where $n = 2$ or 4.

\Ensure 
    \Statex \quad a 64-bit integer hash value $H$

\State write $S$ as $\begin{pmatrix}
    A & B \\
    C & D \\
\end{pmatrix}$, where $A$, $B$, $C$ and $D$ are $n\times n$ binary matrices, $Q\equiv \bigoplus_{i=1}^{\frac{n}{2}} X$
\State $T_1 \leftarrow A Q A^T \, \text{(mod 2)}$
\State $T_2 \leftarrow A Q C^T \, \text{(mod 2)}$
\State $T_3 \leftarrow C Q C^T \, \text{(mod 2)}$
\State pack the $3n^2$ elements (bits) of $T_1$, $T_2$ and $T_3$ to an integer $H$
\end{algorithmic}
\label{algo: hashing_algo}
\end{algorithm}

\begin{algorithm}[H]
\caption{Classifying all of the 2-qubit or 4-qubit Clifford operators}
\begin{algorithmic}[1]

\Require Number of qubits $n\in\{2,4\}$; Clifford operators $\{S_i\}$ represented as binary symplectic matrices (generated on the fly).
\Ensure Number of classes $N_{\rm class}$; class lists $\{L_j\}_{j=1}^{N_{\rm class}}$.

\State $N \leftarrow 2^{n^2} \prod_{j=1}^n (4^j - 1)$ 
\State initialize $N_{\rm class}$ to 0, class lists $L_j$ as empty lists, and $D$ as an empty dictionary.

\For{$i=1$ to $N$}
	\State hash $S_i$ to an integer $H_i$ using Algorithm \ref{algo: hashing_algo} 
    \If{ $H_i$ is in $D$ } 
        \State $j \leftarrow D(H_i)$
        \State let $T$ be an arbitrary matrix in $L_j$
        \State append $S_i$ to the end of $L_j$
    \Else
        \State $N_{\rm class} \leftarrow N_{\rm class} + 1$
        \State $D(H_i) \leftarrow N_{\rm class}$
        \State append $S_i$ to the end of $L_{N_{\rm class}}$
    \EndIf
\EndFor
\end{algorithmic}
\label{algo: classification}
\end{algorithm}

\subsection*{Equivalence-Class Criterion}

All matrix operations in this section are over $\mathbb{Z}_2$.
We assume that $n$ is even; in the applications considered here, $n=2$ or
$n=4$. We use the binary symplectic convention
\begin{equation}
    \begin{gathered}
        S\Omega S^T = \Omega,
        \qquad
        \Omega =
        \begin{pmatrix}
            Q & 0 \\
            0 & Q
        \end{pmatrix},\\
        Q \equiv \bigoplus_{i=1}^{n/2}
        \begin{pmatrix}
            0 & 1 \\
            1 & 0
        \end{pmatrix}.
    \end{gathered}
\end{equation}
The matrix $Q$ is symmetric and satisfies $Q^2=I$. For an $n$-qubit Clifford
action, write
\begin{equation}
    \begin{gathered}
        S =
        \begin{pmatrix}
            A & B \\
            C & D
        \end{pmatrix}
        =
        \begin{pmatrix}
            M & N
        \end{pmatrix},\\
        M =
        \begin{pmatrix}
            A \\
            C
        \end{pmatrix},
        \qquad
        N =
        \begin{pmatrix}
            B \\
            D
        \end{pmatrix}.
    \end{gathered}
\end{equation}
where $A,B,C,D$ are $n\times n$ blocks, and $M$ and $N$ are both
$2n\times n$ block-column matrices. The three hash matrices are
\begin{equation}
    \begin{gathered}
        T_1(S)=AQA^T,\\
        T_2(S)=AQC^T,\\
        T_3(S)=CQC^T.
    \end{gathered}
\end{equation}

\begin{lem}[Block identities]
If $S$ is symplectic, then
\begin{align}
    A Q A^T + B Q B^T &= Q, \label{eq:complete-row-1}\\
    C Q C^T + D Q D^T &= Q, \\
    A Q C^T + B Q D^T &= 0, \label{eq:complete-row-3}\\
    A^T Q A + C^T Q C &= Q, \label{eq:complete-col-1}\\
    B^T Q B + D^T Q D &= Q, \\
    A^T Q B + C^T Q D &= 0. \label{eq:complete-col-3}
\end{align}
Equivalently, in terms of $M$ and $N$,
\begin{equation}
    \begin{gathered}
        M^T\Omega M=Q,\qquad
        N^T\Omega N=Q,\\
        M^T\Omega N=0.
    \end{gathered}
    \label{eq:complete-MN-identities}
\end{equation}
\end{lem}
\begin{proof}
Expanding $S\Omega S^T=\Omega$ gives
\begin{align}
    \begin{pmatrix}
        A & B \\
        C & D
    \end{pmatrix}
    \begin{pmatrix}
        Q & 0 \\
        0 & Q
    \end{pmatrix}
    \begin{pmatrix}
        A^T & C^T \\
        B^T & D^T
    \end{pmatrix}
    =
    \begin{pmatrix}
        Q & 0 \\
        0 & Q
    \end{pmatrix},
\end{align}
which gives Eqs.~\eqref{eq:complete-row-1}--\eqref{eq:complete-row-3}.
Since $S$ is invertible and $\Omega^2=I$, the condition $S\Omega S^T=\Omega$
is equivalent to $S^T\Omega S=\Omega$. Expanding $S^T\Omega S=\Omega$ gives
Eqs.~\eqref{eq:complete-col-1}--\eqref{eq:complete-col-3}, or equivalently
Eq.~\eqref{eq:complete-MN-identities}.
\end{proof}

\begin{lem}[Column-space characterization]
For a symplectic matrix $S=(M\;N)$, the matrix
\begin{align}
    G(S)\equiv M Q M^T
    =
    \begin{pmatrix}
        AQA^T & AQC^T \\
        CQA^T & CQC^T
    \end{pmatrix}
\end{align}
has the same column space as $M$.
Here the lower-left block satisfies $CQA^T=(AQC^T)^T$, since $Q^T=Q$.
\end{lem}
\begin{proof}
Every column of $G(S)=MQM^T$ lies in the column space of $M$. Conversely, by
Eq.~\eqref{eq:complete-MN-identities},
\begin{align}
    G(S)\Omega M
    =
    M Q M^T\Omega M
    =
    M Q Q
    =
    M.
\end{align}
Thus every column of $M$ lies in the column space of $G(S)$. Therefore $G(S)$
and $M$ have the same column space.
\end{proof}

\begin{thm}[Equivalence-class criterion]
Let $S_1$ and $S_2$ be two $n$-qubit Clifford symplectic matrices. Then
$S_1$ and $S_2$ belong to the same equivalence class,
\begin{align}
    S_2^{-1}S_1 =
    \begin{pmatrix}
        s & 0 \\
        0 & s'
    \end{pmatrix},
\end{align}
if and only if
\begin{align}
    T_1(S_1)&=T_1(S_2),\\
    T_2(S_1)&=T_2(S_2),\\
    T_3(S_1)&=T_3(S_2).
\end{align}
\end{thm}

\begin{proof}
We first prove the necessary condition. If $S_1$ and $S_2$ are in the same
equivalence class, then
\begin{align}
    S_1 =
    S_2
    \begin{pmatrix}
        s & 0 \\
        0 & s'
    \end{pmatrix}.
\end{align}
Therefore $M_1=M_2s$. Since the block-diagonal matrix is symplectic, its
upper-left block satisfies $sQs^T=Q$. It follows that
\begin{align}
    M_1 Q M_1^T
    =
    M_2 s Q s^T M_2^T
    =
    M_2 Q M_2^T.
\end{align}
Equating the independent blocks gives
\begin{align}
    A_1 Q A_1^T &= A_2 Q A_2^T,\\
    A_1 Q C_1^T &= A_2 Q C_2^T,\\
    C_1 Q C_1^T &= C_2 Q C_2^T.
\end{align}
Thus the three hash matrices are identical.

We now prove the sufficient condition. Assume
\begin{align}
    A_1 Q A_1^T &= A_2 Q A_2^T,\\
    A_1 Q C_1^T &= A_2 Q C_2^T,\\
    C_1 Q C_1^T &= C_2 Q C_2^T.
\end{align}
Since $Q^T=Q$, the transpose of the second equality gives
$C_1QA_1^T=C_2QA_2^T$. Hence these three equalities are equivalent to
\begin{align}
    M_1 Q M_1^T = M_2 Q M_2^T.
    \label{eq:complete-gram-equality}
\end{align}
By the column-space lemma, $M_iQM_i^T$ and $M_i$ have the same column space.
Eq.~\eqref{eq:complete-gram-equality} therefore implies that $M_1$ and
$M_2$ have the same column space. Moreover, $M_i$ has full column rank because
$M_i^T\Omega M_i=Q$ and $Q$ is invertible. Thus there exists an invertible
$n\times n$ matrix $s$ such that
\begin{align}
    M_1=M_2s.
    \label{eq:complete-M-relation}
\end{align}
Using Eq.~\eqref{eq:complete-MN-identities}, we also obtain
\begin{align*}
    Q
    =
    M_1^T\Omega M_1
    =
    s^T M_2^T\Omega M_2s
    =
    s^TQs,
\end{align*}
such that $s$ is symplectic with respect to $Q$. Equivalently, since $s$ is
invertible, $sQs^T=Q$.

The same argument can be applied to the second block column. By
Eqs.~\eqref{eq:complete-row-1}--\eqref{eq:complete-row-3}, the assumed
equality of the three hash matrices implies
\begin{align}
    N_1 Q N_1^T = N_2 Q N_2^T.
\end{align}
Repeating the column-space argument for $N_i$ and using
$N_i^T\Omega N_i=Q$, we obtain an invertible $n\times n$ matrix $s'$ such that
\begin{align}
    N_1=N_2s'.
\end{align}
Moreover,
\begin{align*}
    Q
    =
    N_1^T\Omega N_1
    =
    {s'}^T N_2^T\Omega N_2s'
    =
    {s'}^TQs',
\end{align*}
so $s'$ is also symplectic with respect to $Q$. Combining the two block-column
relations gives
\begin{align}
    S_1
    =
    \begin{pmatrix}
        M_1 & N_1
    \end{pmatrix}
    =
    \begin{pmatrix}
        M_2s & N_2s'
    \end{pmatrix}
    =
    S_2
    \begin{pmatrix}
        s & 0 \\
        0 & s'
    \end{pmatrix}.
\end{align}
Thus $S_1$ and $S_2$ differ only by independent Clifford actions on the two
halves and belong to the same equivalence class.
\end{proof}

\section*{Acknowledgment}
This work was supported by the Quantum Science and Technology-National Science and Technology Major Project (2023ZD0300200) and the Fundamental Research Funds for the Central Universities. MQ acknowledges the support from the National Natural Science Foundation of China (Grant No. 12522406 and No. 12274290), the Innovation Program for Quantum Science and Technology (2021ZD0301902), and the National Key Research and Development Program of MOST of China (2022YFA1405400).

\section*{Code and data availability}
Source code to reproduce the reported results can be found
at Ref. \cite{campsmodule}.
Some examples can be found at Ref. \cite{campsexample}.

\bibliography{main}

@misc{campsmodule,
    howpublished = {\url{https://github.com/Quantum-Chemistry-Group-BNU/FOCUS/tree/master/pyfocus/camps}},
    year = {2026}
}

@misc{campsexample,
    howpublished = {\url{https://github.com/dilandaer/CAMPS_examples}},
    year = {2026}
}

@article{huang2026augmenting,
  title={Augmenting density matrix renormalization group with matchgates and clifford circuits},
  author={Huang, Jiale and Qian, Xiangjian and Li, Zhendong and Qin, Mingpu},
  journal={Chin. Phys. Lett.},
  volume={43},
  number={4},
  pages={040708},
  year={2026},
  publisher={Chinese Physical Society and IOP Publishing Ltd}
}

@misc{linkToFCIDUMPfe2fe4,
   author = {Li, Zhendong},
   title = {Active-space model for Iron-Sulfur Clusters},
   year = {2026},
   howpublished = {\url{https://github.com/zhendongli2008/Active-space-model-for-Iron-Sulfur-Clusters}},
   note = {Accessed 2026-06-05}
}

@article{li2021expressibility,
  title={Expressibility of comb tensor network states (CTNS) for the P-cluster and the FeMo-cofactor of nitrogenase},
  author={Li, Zhendong},
  journal={Electron. Struct.},
  volume={3},
  number={1},
  pages={014001},
  year={2021},
}

@article{chan2016matrix,
  title={Matrix product operators, matrix product states, and ab initio density matrix renormalization group algorithms},
  author={Chan, Garnet Kin and Keselman, Anna and Nakatani, Naoki and Li, Zhendong and White, Steven R},
  journal={J. Chem. Phys.},
  volume={145},
  number={1},
  pages={014102},
  year={2016},
}

@article{li2017spin,
  title={Spin-projected matrix product states: Versatile tool for strongly correlated systems},
  author={Li, Zhendong and Chan, Garnet Kin-Lic},
  journal={J. Chem. Theory Comput.},
  volume={13},
  number={6},
  pages={2681--2695},
  year={2017},
}

@article{eisert2010colloquium,
  title={Colloquium: Area laws for the entanglement entropy},
  author={Eisert, Jens and Cramer, Marcus and Plenio, Martin B},
  journal={Rev. Mod. Phys.},
  volume={82},
  number={1},
  pages={277},
  year={2010},
}

@article{orus2014practical,
  title={A practical introduction to tensor networks: Matrix product states and projected entangled pair states},
  author={Or{\'u}s, Rom{\'a}n},
  journal={Ann. Phys.},
  volume={349},
  pages={117--158},
  year={2014},
}

@article{chan2011density,
  title={The density matrix renormalization group in quantum chemistry},
  author={Chan, Garnet Kin-Lic and Sharma, Sandeep},
  journal={Annu. Rev. Phys. Chem.},
  volume={62},
  pages={465--481},
  year={2011},
}

@article{barcza2011quantum,
  title={Quantum-information analysis of electronic states of different molecular structures},
  author={Barcza, G and Legeza, {\"O} and Marti, KH and Reiher, M},
  journal={Phys. Rev. A},
  volume={83},
  number={1},
  pages={012508},
  year={2011},
}

@article{olivares2015ab,
  title={The ab-initio density matrix renormalization group in practice},
  author={Olivares-Amaya, Roberto and Hu, Weifeng and Nakatani, Naoki and Sharma, Sandeep and Yang, Jun and Chan, Garnet Kin-Lic},
  journal={J. Chem. Phys.},
  volume={142},
  number={3},
  pages={034102},
  year={2015},
}

@article{jordan1928pauli,
  title={About the Pauli exclusion principle},
  author={Jordan, P and Wigner, Eugene P},
  journal={Z. Phys.},
  volume={47},
  number={631},
  pages={14--75},
  year={1928}
}

@article{Dehaene2003Cliffordgroup,
  title = {Clifford group, stabilizer states, and linear and quadratic operations over GF(2)},
  author = {Dehaene, Jeroen and De Moor, Bart},
  journal = {Phys. Rev. A},
  volume = {68},
  issue = {4},
  pages = {042318},
  numpages = {10},
  year = {2003},
  month = {Oct},
  doi = {10.1103/PhysRevA.68.042318},
  url = {https://link.aps.org/doi/10.1103/PhysRevA.68.042318}
}

@article{Koenig2014Cliffordenumerate,
    author = {Koenig, Robert and Smolin, John A.},
    title = {How to efficiently select an arbitrary Clifford group element},
    journal = {J. Math. Phys.},
    volume = {55},
    number = {12},
    pages = {122202},
    year = {2014},
    month = {12},
    issn = {0022-2488},
    doi = {10.1063/1.4903507},
    url = {https://doi.org/10.1063/1.4903507},
}

@article{mcclean2020openfermion,
  title={OpenFermion: the electronic structure package for quantum computers},
  author={McClean, Jarrod R and Rubin, Nicholas C and Sung, Kevin J and Kivlichan, Ian D and Bonet-Monroig, Xavier and Cao, Yudong and Dai, Chengyu and Fried, E Schuyler and Gidney, Craig and Gimby, Brendan and Gokhale, Pranav and H{\"a}ner, Thomas and Hardikar, Tarini and Havl{\'i}{\v c}ek, Vojt{\v e}ch and Huang, Oscar and Jiang, Zhang and Liu, Matthew and McArdle, Sam and Neeley, Matthew and O'Brien, Thomas and O'Gorman, Bryan and Ozfidan, Isil and Radin, Maxwell D and Romero, Jonathan and Sawaya, Nicolas P D and Setia, Kanav and Sim, Sukin and Steiger, Damian S and Steudtner, Mark and Sun, Qiming and Sun, Wei and Wang, Daochen and Zhang, Fang and Babbush, Ryan},
  journal={Quantum Sci. Technol.},
  volume={5},
  number={3},
  pages={034014},
  year={2020},
}

@article{white_density_1992,
	title = {Density matrix formulation for quantum renormalization groups},
	volume = {69},
	url = {https://link.aps.org/doi/10.1103/PhysRevLett.69.2863},
	doi = {10.1103/PhysRevLett.69.2863},
	number = {19},
	urldate = {2024-09-22},
	journal = {Phys. Rev. Lett.},
	author = {White, Steven R.},
	month = nov,
	year = {1992},
	pages = {2863--2866},
}

@article{huang_nonstabilizerness_2025,
	title = {Nonstabilizerness entanglement entropy: {A} measure of hardness in the classical simulation of quantum many-body systems with tensor network states},
	volume = {112},
	issn = {2469-9926, 2469-9934},
	shorttitle = {Nonstabilizerness entanglement entropy},
	url = {https://link.aps.org/doi/10.1103/gxdn-zwrw},
	doi = {10.1103/gxdn-zwrw},
	language = {en},
	number = {1},
	urldate = {2025-08-04},
	journal = {Phys. Rev. A},
	author = {Huang, Jiale and Qian, Xiangjian and Qin, Mingpu},
	month = jul,
	year = {2025},
	pages = {012425},
}

@article{xiang_distributed_2024,
	title = {Distributed {Multi}-{GPU} \textit{{Ab} {Initio}} {Density} {Matrix} {Renormalization} {Group} {Algorithm} with {Applications} to the {P}-{Cluster} of {Nitrogenase}},
	volume = {20},
	copyright = {https://doi.org/10.15223/policy-029},
	issn = {1549-9618, 1549-9626},
	url = {https://pubs.acs.org/doi/10.1021/acs.jctc.3c01228},
	doi = {10.1021/acs.jctc.3c01228},
	language = {en},
	number = {2},
	urldate = {2025-11-23},
	journal = {J. Chem. Theory Comput.},
	author = {Xiang, Chunyang and Jia, Weile and Fang, Wei-Hai and Li, Zhendong},
	month = jan,
	year = {2024},
	pages = {775--786},
}

@article{qian_augmenting_2024,
	title = {Augmenting {Density} {Matrix} {Renormalization} {Group} with {Clifford} {Circuits}},
	volume = {133},
	url = {https://link.aps.org/doi/10.1103/PhysRevLett.133.190402},
	doi = {10.1103/PhysRevLett.133.190402},
	number = {19},
	urldate = {2024-12-03},
	journal = {Phys. Rev. Lett.},
	author = {Qian, Xiangjian and Huang, Jiale and Qin, Mingpu},
	month = nov,
	year = {2024},
	pages = {190402},
}

@article{huang_clifford_2025,
	title = {Clifford circuits augmented matrix product states for fermion systems},
	volume = {112},
	url = {https://link.aps.org/doi/10.1103/lwwp-6rqk},
	doi = {10.1103/lwwp-6rqk},
	number = {20},
	urldate = {2026-01-06},
	journal = {Phys. Rev. B},
	author = {Huang, Jiale and Qian, Xiangjian and Qin, Mingpu},
	month = nov,
	year = {2025},
	pages = {205106},
}

@article{qian_clifford_2025,
	title = {Clifford {Circuits} {Augmented} {Time}-{Dependent} {Variational} {Principle}},
	volume = {134},
	url = {https://link.aps.org/doi/10.1103/PhysRevLett.134.150404},
	doi = {10.1103/PhysRevLett.134.150404},
	number = {15},
	urldate = {2026-01-06},
	journal = {Phys. Rev. Lett.},
	author = {Qian, Xiangjian and Huang, Jiale and Qin, Mingpu},
	month = apr,
	year = {2025},
	pages = {150404},
}

@article{fu_clifford_2025,
	title = {Clifford augmented density matrix renormalization group for \textit{ab initio} quantum chemistry},
	volume = {112},
	issn = {2469-9950, 2469-9969},
	url = {https://link.aps.org/doi/10.1103/4ng4-vzz6},
	doi = {10.1103/4ng4-vzz6},
	language = {en},
	number = {19},
	urldate = {2025-12-22},
	journal = {Phys. Rev. B},
	author = {Fu, Lizhong and Shang, Honghui and Yang, Jinlong and Guo, Chu},
	month = nov,
	year = {2025},
	pages = {195111},
}

@article{sun_recent_2020,
	title = {Recent developments in the {PySCF} program package},
	volume = {153},
	issn = {0021-9606},
	url = {https://doi.org/10.1063/5.0006074},
	doi = {10.1063/5.0006074},
	number = {2},
	urldate = {2026-01-06},
	journal = {J. Chem. Phys.},
	author = {Sun, Qiming and Zhang, Xing and Banerjee, Samragni and Bao, Peng and Barbry, Marc and Blunt, Nick S. and Bogdanov, Nikolay A. and Booth, George H. and Chen, Jia and Cui, Zhi-Hao and Eriksen, Janus J. and Gao, Yang and Guo, Sheng and Hermann, Jan and Hermes, Matthew R. and Koh, Kevin and Koval, Peter and Lehtola, Susi and Li, Zhendong and Liu, Junzi and Mardirossian, Narbe and McClain, James D. and Motta, Mario and Mussard, Bastien and Pham, Hung Q. and Pulkin, Artem and Purwanto, Wirawan and Robinson, Paul J. and Ronca, Enrico and Sayfutyarova, Elvira R. and Scheurer, Maximilian and Schurkus, Henry F. and Smith, James E. T. and Sun, Chong and Sun, Shi-Ning and Upadhyay, Shiv and Wagner, Lucas K. and Wang, Xiao and White, Alec and Whitfield, James Daniel and Williamson, Mark J. and Wouters, Sebastian and Yang, Jun and Yu, Jason M. and Zhu, Tianyu and Berkelbach, Timothy C. and Sharma, Sandeep and Sokolov, Alexander Yu. and Chan, Garnet Kin-Lic},
	month = jul,
	year = {2020},
	pages = {024109},
}

@misc{QuantumClifford_github,
  author       = {{QuantumSavory}},
  title        = {QuantumClifford.jl},
  year         = {2026},
  howpublished = {\url{https://github.com/QuantumSavory/QuantumClifford.jl}},
  note         = {Accessed 2026-06-05}
}

@article{CASCADE,
  author = {D'Cunha, Ruhee and Crawford, T. Daniel and Motta, Mario and Rice, Julia E.},
  title = {Challenges in the Use of Quantum Computing Hardware-Efficient Ans{\"a}tze in Electronic Structure Theory},
  journal = {J. Phys. Chem. A},
  volume = {127},
  number = {15},
  pages = {3437--3448},
  year = {2023},
  doi = {10.1021/acs.jpca.2c08430}
}

@article{XYZ2F,
author = {Xiao, Xiaoxiao and Zhao, Hewang and Ren, Jiajun and Fang, Wei-Hai and Li, Zhendong},
title = {Physics-Constrained Hardware-Efficient Ansatz on Quantum Computers That Is Universal, Systematically Improvable, and Size-Consistent},
journal = {J. Chem. Theory Comput.},
volume = {20},
number = {5},
pages = {1912-1922},
year = {2024},
doi = {10.1021/acs.jctc.3c00966},

URL = { 
    
        https://doi.org/10.1021/acs.jctc.3c00966
    
    

},
eprint = { 
    
        https://doi.org/10.1021/acs.jctc.3c00966
    
    

}

}

@article{peruzzo2014variational,
  title={A variational eigenvalue solver on a photonic quantum processor},
  author={Peruzzo, Alberto and McClean, Jarrod and Shadbolt, Peter and Yung, Man-Hong and Zhou, Xiao-Qi and Love, Peter J and Aspuru-Guzik, Al{\'a}n and O'Brien, Jeremy L},
  journal={Nat. Commun.},
  volume={5},
  number={1},
  pages={4213},
  year={2014}
}

@article{mcclean2016theory,
  title={The theory of variational hybrid quantum-classical algorithms},
  author={McClean, Jarrod R and Romero, Jonathan and Babbush, Ryan and Aspuru-Guzik, Al{\'a}n},
  journal={New J. Phys.},
  volume={18},
  number={2},
  pages={023023},
  year={2016},
}

@article{sun2024toward,
  title={Toward chemical accuracy with shallow quantum circuits: A Clifford-based Hamiltonian engineering approach},
  author={Sun, Jiace and Cheng, Lixue and Li, Weitang},
  journal={J. Chem. Theory Comput.},
  volume={20},
  number={2},
  pages={695--707},
  year={2024},
}

@article{ren2020general,
  title={A general automatic method for optimal construction of matrix product operators using bipartite graph theory},
  author={Ren, Jiajun and Li, Weitang and Jiang, Tong and Shuai, Zhigang},
  journal={J. Chem. Phys.},
  volume={153},
  number={8},
  pages={084118},
  year={2020},
}

@article{vidal_class_2008,
	title = {Class of {Quantum} {Many}-{Body} {States} {That} {Can} {Be} {Efficiently} {Simulated}},
	volume = {101},
	url = {https://link.aps.org/doi/10.1103/PhysRevLett.101.110501},
	doi = {10.1103/PhysRevLett.101.110501},
	number = {11},
	urldate = {2024-12-30},
	journal = {Phys. Rev. Lett.},
	author = {Vidal, G.},
	month = sep,
	year = {2008},
	pages = {110501},
}

@article{vidal_entanglement_2007,
	title = {Entanglement {Renormalization}},
	volume = {99},
	number = {22},
	doi = {10.1103/PhysRevLett.99.220405},
	journal = {Phys. Rev. Lett.},
	author = {Vidal, G},
	year = {2007},
	pages = {220405},
}

@article{vidal2003efficient,
  title={Efficient classical simulation of slightly entangled quantum computations},
  author={Vidal, Guifr{\'e}},
  journal={Phys. Rev. Lett.},
  volume={91},
  number={14},
  pages={147902},
  year={2003},
}

@article{cirac2021matrix,
  title={Matrix product states and projected entangled pair states: Concepts, symmetries, theorems},
  author={Cirac, J Ignacio and Perez-Garcia, David and Schuch, Norbert and Verstraete, Frank},
  journal={Rev. Mod. Phys.},
  volume={93},
  number={4},
  pages={045003},
  year={2021},
}

@article{ran2020encoding,
  title={Encoding of matrix product states into quantum circuits of one-and two-qubit gates},
  author={Ran, Shi-Ju},
  journal={Phys. Rev. A},
  volume={101},
  number={3},
  pages={032310},
  year={2020},
}

@article{qian2023augmenting,
  title={Augmenting density matrix renormalization group with disentanglers},
  author={Qian, Xiangjian and Qin, Mingpu},
  journal={Chin. Phys. Lett.},
  volume={40},
  number={5},
  pages={057102},
  year={2023},
}

@article{evenbly2009algorithms,
  title={Algorithms for entanglement renormalization},
  author={Evenbly, Glen and Vidal, Guifr{\'e}},
  journal={Phys. Rev. B},
  volume={79},
  number={14},
  pages={144108},
  year={2009},
}

@article{gottesman1998heisenberg,
  title={The Heisenberg representation of quantum computers},
  author={Gottesman, Daniel},
  journal={arXiv Preprint;},
  note={arXiv: quant-ph/9807006},
  year={1998}
}

@book{nielsen2010quantum,
  title={Quantum computation and quantum information},
  author={Nielsen, Michael A and Chuang, Isaac L},
  year={2010},
  publisher={Cambridge university press}
}

@article{aaronson2004improved,
  title={Improved simulation of stabilizer circuits},
  author={Aaronson, Scott and Gottesman, Daniel},
  journal={Phys. Rev. A},
  volume={70},
  number={5},
  pages={052328},
  year={2004},
}

@article{bravyi2002fermionic,
  title={Fermionic quantum computation},
  author={Bravyi, Sergey B and Kitaev, Alexei Yu},
  journal={Ann. Phys.},
  volume={298},
  number={1},
  pages={210--226},
  year={2002},
}

@article{shang2023schrodinger,
  title={Schr{\"o}dinger-Heisenberg variational quantum algorithms},
  author={Shang, Zhong-Xia and Chen, Ming-Cheng and Yuan, Xiao and Lu, Chao-Yang and Pan, Jian-Wei},
  journal={Phys. Rev. Lett.},
  volume={131},
  number={6},
  pages={060406},
  year={2023},
}

@article{mishmash2023hierarchical,
  title={Hierarchical Clifford transformations to reduce entanglement in quantum chemistry wave functions},
  author={Mishmash, Ryan V and Gujarati, Tanvi P and Motta, Mario and Zhai, Huanchen and Chan, Garnet Kin-Lic and Mezzacapo, Antonio},
  journal={J. Chem. Theory Comput.},
  volume={19},
  number={11},
  pages={3194--3208},
  year={2023},
}

@book{gottesman1997stabilizer,
  title={Stabilizer codes and quantum error correction},
  author={Gottesman, Daniel},
  year={1997},
  publisher={California Institute of Technology}
}

@article{li2025entanglement,
  title = {Entanglement-Minimized Orbitals Enable Faster Quantum Simulation of Molecules},
  author = {Li, Zhendong},
  journal = {Phys. Rev. Lett.},
  volume = {135},
  issue = {21},
  pages = {210601},
  numpages = {6},
  year = {2025},
  month = {Nov},
  doi = {10.1103/bwvc-z9hz},
  url = {https://link.aps.org/doi/10.1103/bwvc-z9hz}
}

@article{hehre1969self,
  title={Self-consistent molecular-orbital methods. I. Use of Gaussian expansions of Slater-type atomic orbitals},
  author={Hehre, Warren J and Stewart, Robert F and Pople, John A},
  journal={J. Chem. Phys.},
  volume={51},
  number={6},
  pages={2657--2664},
  year={1969},
}

@article{hehre1972self,
  title={Self-consistent molecular orbital methods. XII. Further extensions of Gaussian-type basis sets for use in molecular orbital studies of organic molecules},
  author={Hehre, Warren J and Ditchfield, Robert and Pople, John A},
  journal={J. Chem. Phys.},
  volume={56},
  number={5},
  pages={2257--2261},
  year={1972}
}

@article{seeley2012bravyi,
  title={The Bravyi-Kitaev transformation for quantum computation of electronic structure},
  author={Seeley, Jacob T and Richard, Martin J and Love, Peter J},
  journal={J. Chem. Phys.},
  volume={137},
  number={22},
  pages={224109},
  year={2012}
}

@article{whitfield2011simulation,
  title={Simulation of electronic structure Hamiltonians using quantum computers},
  author={Whitfield, James D and Biamonte, Jacob and Aspuru-Guzik, Al{\'a}n},
  journal={Mol. Phys.},
  volume={109},
  number={5},
  pages={735--750},
  year={2011}
}

@article{wang2025generalized,
  title={Generalized many-body perturbation theory for the electron correlation energy: Multireference random phase approximation via diagrammatic resummation},
  author={Wang, Yuqi and Fang, Wei-Hai and Li, Zhendong},
  journal={J. Phys. Chem. Lett.},
  volume={16},
  number={12},
  pages={3047--3055},
  year={2025}
}

@article{Hostens2005Stabilizer,
  title={Stabilizer states and Clifford operations for systems of arbitrary dimensions and modular arithmetic},
  author={Hostens, Erik and Dehaene, Jeroen and De Moor, Bart},
  journal={Phys. Rev. A},
  volume={71},
  number={4},
  pages={042315},
  year={2005}
}

@article{gxdn-zwrw,
  title = {Nonstabilizerness entanglement entropy: A measure of hardness in the classical simulation of quantum many-body systems with tensor network states},
  author = {Huang, Jiale and Qian, Xiangjian and Qin, Mingpu},
  journal = {Phys. Rev. A},
  volume = {112},
  issue = {1},
  pages = {012425},
  numpages = {11},
  year = {2025},
  month = {Jul},
  publisher = {American Physical Society},
  doi = {10.1103/gxdn-zwrw},
  url = {https://link.aps.org/doi/10.1103/gxdn-zwrw}
}

@book{xiang2023density,
  url={https://books.google.com/books?hl=en&lr=&id=E5fxEAAAQBAJ&oi=fnd&pg=PP1&ots=Hqdq-ApAx0&sig=xi-IvwDkPLk7CTWPcB_d5zRtFZk},
  title={{Density Matrix and Tensor Network Renormalization}},
  author={Xiang, Tao},
  year={2023},
  publisher={Cambridge University Press}
}
\end{document}